\documentclass[12pt]{article}
\usepackage[pdftex]{graphicx}
\usepackage{xcolor}
\usepackage{amsmath}
\usepackage{booktabs}
\usepackage{authblk}
\usepackage{apacite}
\usepackage{longtable}
\usepackage{amsfonts}
\usepackage{mathrsfs}
\usepackage{amsmath,amsthm,amssymb}
\usepackage{caption}
\usepackage{subcaption}

\small\normalsize
\newtheorem{proposition}{Proposition}
\newtheorem{assumption}{Assumption}

\newtheorem{remark}{Remark}

 \def\E{\mbox{E}}

 \def\bzeta{\mbox{\boldmath $\zeta$}}
\def\bet{\mbox{\boldmath $\eta$}} 
\def\bGamma{\mbox{\boldmath $\Gamma$}}
 
\def\bxi{\mbox{\boldmath $\xi$}}

\def\bpsi{\mbox{\boldmath $\psi$}}

\def\bchi{\mbox{\boldmath $\chi$}}

\def\bef{\mathbf{f}}

\def\bvartheta{\mbox{\boldmath $\vartheta$}}

\def\bSigma{\mathbf{\Sigma}}

\def\bphi{\mbox{\boldmath $\phi$}}
   \def\bA{\mathbf{A}}
 \def\bK{\mathbf{K}} \def\bQ{\mathbf{Q}} 
\def\bF{\mathbf{F}} \def\bB{\mathbf{B}} \def\bD{\mathbf{D}} 
  \def\bp{\mathbf{p}} \def\bP{\mathbf{P}}
 
\def\bJ{\mathbf{J}}  \def\bL{\mathbf{L}} 
\def\bu{\mathbf{u}}    \def\0{\mbox{\bf{0}}}
\def\bs{\mathbf{s}}\def\bS{\mathbf{S}}  \def\bI{\mathbf{I}}
\def\bz{\mathbf{z}}\def\bZ{\mathbf{Z}}   \def\bM{\mathbf{M}}
    
  \def\bx{\mathbf{x}}
             \def\WN{\mbox{WN}}
\def\E{\mbox{E}}       \def\diag{\mbox{diag}}
             \def\Cov{\mbox{Cov}} \def  \Var{\mbox{Var}}                                          
                
\def\r{\mbox{$r$}}

\graphicspath{{AltreFigure/}}
\begin{document}
\bibliographystyle{apacite}
\renewcommand{\BBAA}{and}
\renewcommand{\BBAB}{and}
\renewcommand{\BRetrievedFrom}{Downloadable at }

\clearpage

\title{{\huge Band-Pass Filtering with High-Dimensional Time Series} \thanks{ \scriptsize The authors are grateful to Esther Ruiz, Pilar Poncela and Valentina Aprigliano for their valuable comments, which led to several improvements in both the presentation and the content of the paper. The authors are also grateful to the partecipants of 41th International Symposium on Forecasting 2021; 3rd Italian Workshop of Econometrics and Empirical Economics, Rimini 2022; 5th Workshop on High-Dimensional Time Series in Macroeconomics and Finance, Wien 2022; 42th International Symposium on Forecasting 2022, Oxford; Rome-Waseda Time Series Symposium Villa Mondragone 2022; 16th Computational and Financial Econometrics (CFE), King's College London 2022. 
Alessandro Giovannelli and Tommaso Proietti gratefully acknowledge financial support by the Italian Ministry of Education, University and Research, Progetti di Ricerca di Interesse Nazionale, research project 2020-2023, project 2020N9YFFE.}}

\author[1]{A. Giovannelli}
\author[2]{M. Lippi}
\author[3]{T. Proietti}
\affil[1]{\small University of L'Aquila, Italy}
\affil[2]{\small Einaudi Institute for Economics and Finance, Rome, Italy}
\affil[3]{\small University of Rome ``Tor Vergata", Italy}

\date{May 8 2023}
\maketitle
\begin{abstract}
The paper deals with the construction of a synthetic indicator of economic growth, obtained by projecting a quarterly measure of aggregate economic activity, namely gross domestic product (GDP), into the space spanned by a finite number of smooth principal components, representative of the medium-to-long-run component of economic growth of a high-dimensional time series, available at the monthly frequency.
The smooth principal components result from applying a  cross-sectional filter distilling the low-pass component of growth in real time.
The outcome of the projection is a monthly nowcast of the medium-to-long-run component of GDP growth. After discussing the theoretical properties of the indicator, we deal with the assessment of its reliability and predictive validity with reference to a panel of macroeconomic U.S. time series.
\vspace{.5cm}


\vspace{.5cm}
\end{abstract}
\thispagestyle{empty}
\newpage

\section{Introduction}


Nowcasting deals with the real time assessment of the underlying growth rate of the economy, with the aid of the available indirect information provided by a large number of more timely macroeconomic monthly indicators.
A broad and representative measure of aggregate economic activity is offered by  gross domestic product (GDP), which is however available only quarterly and with a publication delay. On the contrary, it is desirable to carry out the above assessment at least with monthly frequency.

There is a large literature on the GDP nowcasting problem. Bridge models, see \shortciteA{baffigi2004bridge} and MIxed frequency DAta Sampling models   \shortcite{ghysels2007midas,kuzin2013pooling,foroni2014comparison}, aim at nowcasting quarterly GDP using monthly indicators.
Rather than targeting quarterly GDP, it can be deemed more relevant, for the purpose of monitoring the current state of the economy, to nowcast GDP at the monthly frequency.
Low-dimensional exact dynamic factor models have been proposed by \shortciteA{mariano2003new},  \shortciteA{camacho2010introducing},  \shortciteA{frale2011euromind}, and \shortciteA{aruoba2016improving} for this task. More recent contributions have looked at new 
High-dimensional approximate factor models \shortcite{forni2000generalized,stock2002forecasting,stock2002macroeconomic,bai2002determining,hallin2020time} play a pivotal role for nowcasting. They are based on a solid representation theory and provide the way of distilling the comovements in a large set of macroeconomic time series, without incurring in the curse of dimensionality. For applications to nowcasting see \shortciteA{giannone2008nowcasting}, \shortciteA{Altissimo2010new} and \shortciteA{banbura2013now}, among others.

This paper considers the problem of constructing monthly indicators of the medium-to-long-run (M2LR) component of GDP growth at different horizons, namely quarter-on-quarter and year-on-year growth.  The reference to the  M2LR  component of economic growth is meant to exclude high-frequency uninteresting variation, consisting of fluctuations with periodicity less than or equal to one year.

In the unrealistic setting in which a doubly infinite sample of GDP growth measure is available at the monthly frequency, the assessment of the M2LR component of economic growth can be made by the ideal band-pass filter.  The difficulty lies not only in the fact that monthly GDP is not observed; an  obvious limitation of the ideal  band-pass filter is its two-sided nature; more precisely, it is an  infinite, two-sided, symmetric moving average of the time series. \shortciteA{baxter1999measuring}, who popularized band-pass filtering in macroeconomics,  derive an approximation to the ideal band-pass filter with finite impulse response, which however leaves unaddressed estimation in real time. \shortciteA{christiano2003band} have addressed this problem by proposing forecast extensions, i.e., by replacing the missing future observations with conditional forecasts.  This paper argues that forecasts extensions are not needed, and that the relevant information concerning the M2LR component of economic growth can be distilled in real time from a large cross-section of macroeconomic time series, by extracting smooth common factors and projecting GDP growth on the space spanned by them.

We build our approach on \shortciteA{Altissimo2010new}, who made an essential contribution to the above research question through the development of the New EuroCoin (NE) method. NE is a real-time estimate of the euro area GDP growth, cleaned of short-run oscillations. More precisely the ideal target for NE is the component of the GDP growth rate obtained by removing the fluctuations of a period shorter than or equal to 1 year. NE avoids the end-of-sample bias typical of two-sided filters by projecting the target onto suitable linear combinations of a large set of monthly variables. Such linear combinations are designed to discard irrelevant information (idiosyncratic and high-frequency noise) and retain relevant information (common, cyclical and long-run waves).

Against this background, this paper makes several contributions to the above literature, which can be regarded as a development and a refinement of \shortciteA{Altissimo2010new}. From the methodological standpoint, we derive the estimators of the common component and the common factors from best linear prediction principles. In particular, we show that in a finite cross-section  the minimum mean square linear estimators of the components depend on data via a finite number of cross-sectional averages of the original series, represented by the generalized principal components of the respective covariance matrices.  As the dimension of the cross-section diverges, the estimation error vanishes.
The new results enable to clarify some aspects of the theory, among which the relation between the standard \shortcite{stock2002forecasting} and generalized principal components \shortcite[Section 3.3]{forni2005generalized}. 

Secondly, while NE focuses only on quarterly growth horizon, the methodology is extended so as to consider all possible horizons for growth; in particular, we think that making available a measure of underlying growth at an annual rate is relevant for practitioners.

From the empirical standpoint, we provide a novel application dealing with nowcasting the M2LR component of U.S. GDP, using the popular FRED-MD dataset \shortcite{mccracken2016fred}, consisting of 122 monthly macroeconomic time series, for estimating  the smooth common factors. The outcome is a monthly indicator of underlying growth
that we label US COIN (U.S. COincident INdicator).
We then address the question of evaluating our indicator, by assessing the accuracy of the nowcasts of the M2LR component of GDP growth. This is not straightforward, as the target measure is not observable. Thus, we provide a discussion concerning the construction of an \emph{oracle} measure that we elect as the nowcast target.
A comparative validation of US COIN vis-a-vis alternative estimation strategies is then performed, on the basis of a pseudo-real time nowcasting exercise.
Finally, we propose a bootstrap procedure for evaluating the finite sample estimation uncertainty of the indicator.

The plan of the paper is the following. The next section reviews the specification of the dynamic factor model and outlines the nowcasting problem.  Section \ref{sec:mmsle} deals with optimal linear estimation of the common component and its M2LR component, and establishes the consistency of the estimator. In section \ref{sec:usacoin} the smooth factors are used to construct the indicator of M2LR growth at the monthly frequency. The estimation of the spectral density of the series and of the components is considered in section \ref{sec:est}, where we propose a parametric bootstrap method for assessing the finite-sample estimation error uncertainty.
Our US COIN indicator is presented in section \ref{sec:predval}, which provides details on the dataset and the construction of the benchmark measure of M2LR growth used to assess the methodology.
Finally, in section \ref{sec:concl} we draw our conclusions.


\section{The Dynamic Factor Model \label{sec:dfm}}

Let $\{y_t, t\in \mathbb{Z}\}$ denote the logarithm of GDP in month $t$ and let $\Delta y_t = y_t-y_{t-1}$ denote its monthly growth rate. We do observe  neither $y_t$  nor $\Delta y_t$, but quarterly GDP, $\{Y_\tau, \tau = 1, 2, \ldots\}$, is available every third month,  at the end of times $3\tau$, $\tau = 1, 2, \ldots$. Its logarithm, $y_\tau = \ln Y_\tau$, is an approximation to the sum of the logarithms of three consecutive months,  $y_\tau \simeq y_{3\tau}+y_{3\tau-1}+y_{3\tau-2}$, where the left hand side is a systematic sample of the process $(1+L+L^2)y_t$, so that every third value is observed. Here $L$ denotes the lag operator, $L^ky_t=y_{t-k}$, and $\Delta = 1-L$.

The rationale is that the M2LR component of $\Delta y_t$ can be retrieved from  its projection on a fixed, but unknown, number of smooth common factors, extracted from a high-dimensional cross-section of monthly time series. The factors are the drivers of the M2LR component of an $n$-dimensional time series
$\{\bx_{t}, t\in \mathbb{Z}\}$, which is assumed to  be  covariance stationary for any $n$.
The projection raises a problem of temporal disaggregation, as   we observe the logarithmic change
$$\begin{array}{lll}
g_\tau &=& y_{\tau}-y_{\tau-1} \\
&\simeq& (y_{3\tau}+y_{3\tau-1}+y_{3\tau-2})-(y_{3\tau-3}+y_{3\tau-4}+y_{3\tau-5}), \\
\end{array}$$
where the left hand side is a systematic sample of $(1+L+L^2)^2 \Delta y_t$.

The purpose of this section is to expose the methodology superintending the construction of the indicator of the M2LR component  of GDP growth. The next subsection deals with the  dynamic factor model formulated at the monthly frequency for the vector $\bx_t$, whose elements, without loss of generality are standardized so that $\E(x_{it})=0$ and $\Var(x_{it})=1,$ where $x_{it}$ is the $i$-th element of the vector $\bx_t$, $i= 1, \ldots, n$.

Further, we denote the cross-covariance function of $\bx_t$ by $\bGamma_k^x = \E(\bx_{t}\bx_{t-k}')$, and the spectral density matrix by   $\bSigma^x(\theta)=\frac{1}{2\pi}\sum_{k=-\infty}^\infty \bGamma_k^x e^{-\imath \theta k}$, where $\theta \in [-\pi, \pi]$ is the frequency in radians. 
\subsection{Specification}

The approximate dynamic factor model  represents $\bx_t$ as the sum of two orthogonal components,
\begin{equation}\label{eq:dfm}
\bx_t = \bchi_t+\bxi_t,
\end{equation}
where  $\bchi_t$ and $\bxi_t$ denote respectively the $n$ dimensional common and idiosyncratic components, with  $\E(\bchi_t)=\0$, $\E(\bxi_t)=\0$, $\E(\bchi_{t}\bchi_{t-k}') = \bGamma_k^\chi$, $\E(\bxi_{t}\bxi_{t-k}') = \bGamma_k^\xi$ and $ \E(\bchi_{t}\bxi_s') =\0, \forall t,s \in \mathbb{Z}$. The idiosyncratic processes  $\{\xi_{it}\}$ are weakly dependent across the cross-sectional dimension. The above representation implies the following additive decomposition of the cross-covariance function, $\bGamma_k^x = \bGamma_k^\chi+\bGamma_k^\xi$. As for the spectral density matrix:
$\bSigma^x(\theta)=\bSigma^\chi(\theta)+\bSigma^\xi(\theta),$
where we have defined $\bSigma^x(\theta)=\frac{1}{2\pi}\sum_{k=-\infty}^\infty \bGamma_k^x e^{-\imath \theta k}$,  $\bSigma^\chi(\theta)=\frac{1}{2\pi}\sum_{k=-\infty}^\infty \bGamma_k^\chi e^{-\imath \theta k}$ and $\bSigma^\xi(\theta)=\frac{1}{2\pi}\sum_{k=-\infty}^\infty \bGamma_k^\xi e^{-\imath \theta k}$.

Let $\lambda_j^x(\theta), z=x,\chi, \xi,$  denote the  $j$-th largest dynamic eigenvalue of $\bSigma^x(\theta)$ and let  $\bp_j^x(\theta)$ be the corrisponding eigenvectors, such that  $\bSigma^x(\theta)\bp_j^x(\theta)=\lambda_j^x(\theta)\bp_j^x(\theta)$, $j=1,\ldots, n$, and the eigendecomposition
$\bSigma^x(\theta) = \sum_{j=1}^n \lambda_j^x(\theta)\bp_j^x(\theta)\bar{\bp}_j^x(\theta)$  holds, where   $\bar{\bp}_j^{x}$ is the complex conjugate transpose of $\bp_j^x$.
Similarly, consider the eigen-decompositions  $\bSigma^\chi(\theta) = \sum_{j=1}^n \lambda_j^\chi(\theta)\bp_j^\chi(\theta)\bar{\bp}_j^\chi(\theta)$,   $\bSigma^\xi(\theta) = \sum_{j=1}^n \lambda_j^\xi(\theta)\bp_j^\xi(\theta)\bar{\bp}_j^\xi(\theta)$.

We also denote with  $\mu_j^x$, $\mu_j^\chi$ and $\mu_j^\xi$ the $j$-th largest eigenvalue of $\bGamma_0^{x}$ $\bGamma_0^{\chi}$ and $\bGamma_0^{\xi}$, respectively, and define $\bS = \left[ \bs_1, \ldots, \bs_j, \ldots, \bs_n\right]$ the $n\times n$ orthogonal matrix whose columns are the eigenvectors of $\bGamma_0^x$, $\bGamma_0^x\bs_j = \mu_j^x \bs_j$. With obvious notation the eigendecomposition of the covariance matrix of $\bx_t$  is $\bGamma_0^x = \bS \bM_x \bS'$, $\bM_x = \diag\{\mu_{1}^x, \ldots, \mu_{n}^x\}$.
Finally, denote  $\bM_\chi = \diag\{\mu_{1}^\chi, \ldots, \mu_{n}^\chi\}$.

As in \shortciteA{forni2005generalized} we make the following assumption.
\begin{assumption}\label{ass:1}
The common component, $\bchi_t$, is characterized by strong cross-sectional dependence, and  the idiosyncratic component, $\bxi_t$, by weak cross-sectional dependence in the sense specified below:
\begin{description}
\item[a)] $\lambda_q^\chi(\theta)\rightarrow \infty $ as $n\rightarrow \infty$, $\theta$-almost everywhere (a.e.) in $[-\pi,\pi]$
\item[b)] $\lambda_{j+1}^\chi(\theta)<\lambda_j^\chi(\theta)$, $j = 1, \ldots, q$, $\theta$-a.e. in $[-\pi,\pi]$.
\item[c)] The eigenvalues of $\bSigma^\xi(\theta)$ are uniformly bounded, i.e., $\lambda_{ 1}^\xi(\theta)\leq K$, for positive real  $K$, for all $n\in \mathbb{N}$ and for all $\theta \in [-\pi,\pi]$ .
\item[d)] $\mu_{j+1}^\chi< \mu_{j}^\chi, j = 1, \ldots, r$, and $\mu_r^\chi\rightarrow \infty$ as $n\rightarrow \infty$.
\item[e)] $0<\mu_{n}^\xi<K<\infty$ as $n\rightarrow \infty$.
\end{description}
\end{assumption}
Assumptions a) and c) identify the common component as a dynamic linear combination of $q$ common factors, where $q$ is the number of diverging eigenvalues, and the idiosyncratic component, which is both serially and contemporaneously weakly dependent. Assumption b) requires that the first $q+1$ eigenvalues are distinct. The space spanned by the common component has finite dimension $r$.

Under Assumption \ref{ass:1} it can be shown \shortcite{forni2000generalized,forni2005generalized,forni2009opening} that $\bx_t$ has a static factor model representation,  driven by $r$ common factors. This follows from the representation of the common component,
$\bchi_t = \bB(L)\bu_t$, where $\bB(L) = \bL [\bD(L)]^{-1}\bK$, $\bL$ is an $n\times r$ matrix of factor loadings, $\bD(L) = \bI - \sum_{j=1}^p \bD_j L^j$ is an $r \times r$ matrix lag polinomial with roots outside the unit circle, $\bK$ is $r\times q$, and $\bu_t\sim \WN(\0, \bI_q)$ are the common shocks, also referred to as the dynamic factors.

\section{Minimum mean square linear estimation of the low-pass component of economic growth \label{sec:mmsle}}

\subsection{Estimation of the common component}

The following proposition derives the minimum mean square linear estimator of the common component based on a finite cross-section.

\begin{proposition}\label{prop:1}
Let $\bQ_\chi = \bM_x^{-\frac{1}{2}}\bS'\bGamma_0^\chi\bS\bM_x^{-\frac{1}{2}}$ have eigendecomposition
$\bQ_\chi = \sum_{j=1}^n \mu_j^{\chi *} \bz_j^* \bz_j^{*\prime}$ and let
$\bZ_\chi^* = [\bz_1^*, \ldots, \bz_r^*]$.
The finite $n$, rank $r$, optimal linear estimator of the common component,  $ \tilde{\bchi}_t$, based on $\bx_t$ is
\begin{equation}
\tilde{\bchi}_t = \bGamma_0^x\bZ_\chi \bM^*_\chi \bZ_\chi'\bx_t,
\label{eq:chihat_all}
\end{equation}
where $\bZ_\chi = \bS\bM_x^{-\frac{1}{2}}  \bZ_\chi^*$ and  $\bM_{\chi *} = \diag\{\mu_{1}^{\chi *}, \ldots, \mu_{r}^{\chi *}\}$.
The linear combinations $\bZ_\chi'\bx_t$ are the $r$ generalized principal components \shortcite{forni2005generalized}.

The mean square estimation error is
\begin{equation}
\E\{(\bchi_t-\tilde{\bchi}_t)(\bchi_t-\tilde{\bchi}_t)' \}= \bGamma_0^\chi-\bGamma_0^\chi \bZ_\chi\bZ_\chi'\bGamma_0^\chi.
\label{eq:msechi}
\end{equation}
\end{proposition}
\begin{proof}
See Appendix \ref{app:prop1}
\end{proof}
Note that minimizing the mean square error matrix means minimizing its operator  norm: if $\bA$ is an $n\times m$ matrix, $\|\bA\|_\Gamma$ is the largest singular value of $(\bGamma_0^{x})^{-1/2'}\bA$.
The columns of the matrix $\bZ_\chi$ are formed by the first $r$  generalized eigenvectors of the two matrices  $\bGamma_0^\chi$ and $\bGamma_0^x$, satisfying
$\bZ_{\chi}'\bGamma_0^\chi\bZ_\chi = \bM_\chi^{*}$ and $\bZ_{\chi}'\bGamma_0^x\bZ_\chi = \bI_r$.  Equation (\ref{eq:chihat_all}) is equivalently written $\tilde{\bchi}_t = \bGamma_0^\chi\bZ_\chi \bZ_\chi'\bx_t$.

The next proposition shows the consistency, as $n\rightarrow\infty$, of the estimator of the common component.
\begin{proposition}\label{prop:2}
$\tilde{\bchi}_t\rightarrow_P \bchi_t$, as $n\rightarrow \infty$    and
$$\lim\limits_{n\rightarrow \infty} \tilde{\bchi}_t = \bGamma_0^x\bZ_\chi  \bZ_\chi'\bx_t.$$
\end{proposition}
\begin{proof}
See Appendix \ref{app:prop2}
\end{proof}
The asymptotic approximation of the estimator of the common component can be written
\begin{equation}
\begin{array}{lll}
\tilde{\bchi}_t &=& \bA_\chi \bef_t\\
&=& \bP_\chi\bx_t
\end{array}\label{eq:chihat}
\end{equation}
where $\bef_t=\bZ_\chi' \bx_t$ are the generalized principal components, and we denoted the loadings matrix $\bA_{\chi} = \bGamma_0^x \bZ_\chi$, while $\bP_{\chi} = \bGamma_0^x \bZ_\chi  \bZ_{\chi}'$.
The space of the common components is spanned by the $r$  columns of $\bA_\chi=\bGamma_0^x \bZ_\chi.$ The vectors are orthonormal in $\mathcal{M}_x^n$, the $n$-dimensional Hilbert space endowed with the inner product $ \|\bx\|_\Gamma  = (\bx'[\bGamma_0^x]^{-1}\bx)^{1/2}$, as $\bA_\chi'[\bGamma_0^x]^{-1}\bA_\chi = \bI_r$. The null space is spanned by $\bGamma_0^x \bZ_\xi$, where $\bZ_{\xi}$ are the $n-r$ generalized eigenvectors satisfying $\bZ_{\xi}'\bGamma_0^x\bZ_\chi = \0$.

The matrix $\bP_\chi$ is the $n \times n$ projection matrix with rank $r$, projecting $\bx_t$ into the space of the common component. The matrix is idempotent, $\bP_\chi^2 = \bP_{\chi}$, and it is not symmetric. 
The $r\times n$ matrix $\bZ_\chi'$ is the weak, or generalized, inverse of $\bA_\chi$, satisfying the two properties:
(i) $\bA_\chi \bZ_\chi'\bA_\chi = \bA_{\chi}$  and (ii) $\bZ_\chi'\bA_\chi \bZ_\chi'= \bZ_\chi'$. Hence, we can write, $\bA_{\chi}^- = \bZ_{\chi}'$, $\bP_{\chi} = \bA_\chi\bA_\chi^-$. See  \shortciteA{rao1974projectors} and \shortciteA[ch. 7]{seber2008matrix}.
According to the terminology in \shortciteA{seber2008matrix}, $\bZ_\chi'$ is a minimum norm reflexive $g$-inverse. It is also a least square generalized inverse, see \shortciteA{rao1971generalized}.

It is useful to compare our estimator with the standard principal components analysis (PCA) estimator of $\bchi_t$ \shortcite{stock2002forecasting, stock2002macroeconomic}.
The latter arises from the projecting of $\bx_t$ on the space of the common component in a metric space endowed with the inner product $ \|\bx\|  = (\bx'\bx)^{1/2}$.
Let $\bS_r$ denote the $n\times r$ matrix whose columns are the $r$ eigenvectors of $\bGamma_0^x$ corresponding to the largest eigenvalues. The PCA
estimator is $\hat{\bchi_t}^{_{PCA}} = \bS_r\bS_r'\bx_t$. Denoting $\bM_r = \diag(\mu_1, \ldots, \mu_r)$,
$$\begin{array}{lll}\Var(\hat{\bchi_t}^{_{PCA}} )&=& \bS_r\bM_r\bS_r' \\
&=& \bGamma_0^x\bS_r\bM_r^{-1}\bS_r'\bGamma_0^x,\end{array}$$
where the second expression is for comparison with $$\Var(\tilde{\bchi_t}) = \bGamma_0^x \bZ_\chi  \bM_\chi^{*2} \bZ_\chi'\bGamma_0^x.$$

The estimation mean square error matrix is
\begin{equation}
\E\{(\bchi_t-\hat{\bchi}_t^{_{PCA}})(\bchi_t-\hat{\bchi}_t^{_{PCA}})' \}= \bGamma_0^\chi(\bGamma_0^{\chi -}- \bS_r \bM_r^{-1}\bS_r')\bGamma_0^\chi,
\label{eq:msepca}
\end{equation}
which shows that $\bS_r \bM_r^{-1}\bS_r'$ is an alternative Moore-Penrose inverse of $\bGamma_0^\chi$ and that the two solutions are asymptotically equivalent.
\begin{remark}
The generalized principal component estimator (\ref{eq:chihat_all}) is the minimum mean square linear estimator in the $\mathcal{M}_x^n$ Hilbert space, in which the mean square error matrix  is defined as
$\E\left(\|\bchi_t - \tilde{\bchi}_t\|_\Gamma^2 \right)=\bGamma_0^{x-1/2'}\E\left\{\left(\bchi_t - \tilde{\bchi}_t\right)\left(\bchi_t - \tilde{\bchi}_t\right)'\right\} \bGamma_0^{x-1/2},$
whereas the standard PCA solution provides the minimum mean square linear estimator in the Hilbert space endowed with the norm $\|\bx\| = \sqrt{\bx'\bx}$.  As a result, in comparing
(\ref{eq:msechi}) and (\ref{eq:msepca}) we have that
 $\E\{(\bchi_t-\tilde{\bchi}_t)(\bchi_t-\tilde{\bchi}_t)' \}-\E\{(\bchi_t-\hat{\bchi}_t^{_{PCA}})(\bchi_t-\hat{\bchi}_t^{_{PCA}})' \}\geq 0$ and
$\E\left(\|\bchi_t - \tilde{\bchi}_t\|_\Gamma^2 \right)-\E\left(\|\bchi_t - \hat{\bchi}_t^{_{PCA}} \|_\Gamma^2 \right)\leq 0$.
Hence, the relative efficiency of the generalized PCA estimator versus the standard PCA estimator depends on the Hilbert space that is assumed.
\end{remark}

\subsection{Smooth generalized principal components: a cross\--sectional low\--pass filter \label{sec:smooth}}

The $r$ generalized principal components  distill the co-movements of the $n$ macroeconomic time series that are pervasive. However, their variability contains high frequency trigonometric components that do not contribute to the M2LR growth component of GDP.
\shortciteA{Altissimo2010new} consider the additive decomposition of the common component into a low-pass component, $\bphi_t$, a high-pass component, $\bpsi_t$,
\begin{equation}
\bchi_t = \bphi_t+ \bpsi_t, \label{eq:dec}
\end{equation}
where $\E(\bphi_t) = \0$,  and  $\E(\bphi_t\bphi_{t-k}')=\bGamma_k^\phi$. The cross-covariance function of $\bphi_t$  is defined as
$$
\bGamma_k^\phi = \sum_{k=1}^q \int_{-\theta_c}^{\theta_c} \lambda_k^x(\theta) \bp_k^x(\theta)\bar{\bp}_k^x(\theta) e^{\imath \theta k} d\theta,
$$
and  $0<\theta_c<\pi$ is the cut-off frequency, e.g., $\theta_c = \pi/6$. The high-pass component is assumed orthogonal to $\bphi_t$ and its cross-covariance function is
$\bGamma_k^\psi = \bGamma_k^\chi - \bGamma_k^\phi$.

The orthogonal decomposition in (\ref{eq:dec}) is achieved via the following orthogonal decomposition of the common shocks of the representation $\bchi_t = \bB(L)\bu_t$:
$$\bu_t = \frac{(1+L)^s}{\varphi(L)}\bet_t + \sqrt{\varsigma}\frac{(1-L)^s}{\varphi(L)}\bzeta_t, $$
where $\bet_t\sim\WN(\0,\bI_q), \bzeta_t\sim\WN(\0, \bI_q), \E(\bet_t\bzeta_t')=\0,$   $\varphi(L)$ is a scalar polynomial satisfying
$$|\varphi(z)|^2 = |1+z|^{2s}+\varsigma |1-z|^{2s}, |z|\geq 1,$$
and $\varsigma$ is related to the cutoff frequency $\theta_c$ by $\varsigma = \left(\frac{1+\cos\theta_c}{1-\cos\theta_c}\right)^s.$
See \shortciteA{proietti2008model} for details. The common shocks are therefore decomposed into a low-pass component and a high-pass one. The former has spectral density
$w_\phi(\theta)\bI_q$, where $$w_\phi(\theta) =\frac{|1+e^{-\imath \theta}|^{2s}}{|\varphi(e^{-\imath \theta})|^2},$$
This decreases monotonically from 1 to 0 as $\theta$ goes from 0 to $\pi$, and taking the value $1/2$ at the frequency $\theta_c$.
As $s\rightarrow \infty$, it tends to the ideal low-pass box-car spectrum $I(|\theta|<\theta_c)$, where $I(\cdot)$ is the indicator function.

Then, the low-pass and high-pass components of $\bchi_t$ are respectively defined as
$$\bphi_t = \bB(L) \frac{(1+L)^s}{\varphi(L)}\bet_t, \;\;\;\;\bpsi_t = \sqrt{\varsigma}\bB(L) \frac{(1-L)^s}{\varphi(L)}\bzeta_t.$$
The former would be estimated on the basis of $\bchi_t$ as $w_\phi(L) \bchi_t$ \shortcite{whittle1963prediction}; here $w_\phi(z) = \frac{|1+z|^{2s}}{|\varphi(z)|^2}$ represents  the optimal Wiener-Kolmogorov filter  \shortcite{whittle1963prediction} for estimating $\bphi_t$. An approximation to the target low-pass component could be obtained by applying the filter for $s=2$ to (\ref{eq:chihat}). This can be done by applying the state space methodology, as in \shortciteA{proietti2008model}, see also Appendix \ref{sec:mbfilt}. However, for $s$ large (e.g., larger than 6) the estimates become computationally unstable.

We hereby show that rather than applying a two sided filter to the estimated common component, we can estimate $\bpsi_t$ by cross-sectional averaging.
In the sequel we denote  the $j$-th largest eigenvalue of $\bGamma_0^\phi$ and $\bGamma_0^\psi$ respectively by  $ \mu_{j}^\phi$ and $\mu_j^{\psi}$. 
The following assumption enables the identification of $\bphi_t$ as $n\rightarrow \infty$.
\begin{assumption} \label{ass:2}
Let $0<r_\phi \leq r$  and $0<\r_{\psi}\leq r$, $r_\phi+r_\psi = r$, be such that $\mu_{r_\phi}^\phi\rightarrow \infty$ and $\mu^\psi_{r_\psi}\rightarrow \infty$ as $n\rightarrow \infty$.
\end{assumption}

The space of the low-pass component is a subspace of that of the common component. The following proposition characterizes the estimator of the low-pass component.

\begin{proposition}\label{prop:3}

Let $\bQ_\phi = \bM_x^{-\frac{1}{2}}\bS'\bGamma_0^\phi\bS\bM_x^{-\frac{1}{2}}$ have eigendecomposition
$\bQ_\phi = \sum_{j=1}^n \mu_{j}^{\phi *} \bz_{\phi j}^{*} \bz_{\phi j}^{*'}$ and let
$\bZ_\phi^* = [\bz_{\phi 1}^*, \ldots, \bz_{\phi, r_\phi}^*]$.
The finite $n$, rank $r_\phi$, optimal linear estimator of the common component,  $ \tilde{\bphi}_t$, based on $\bx_t$, is
\begin{equation}
\tilde{\bphi}_t = \bGamma_0^x\bZ_\phi \bM^*_\phi \bZ_\phi'\bx_t,
\label{eq:chihat_low}
\end{equation}
where $\bZ_\phi = \bM_x^{-\frac{1}{2}}\bS' \bZ_\phi^*$ and  $\bM_{\phi *} = \diag\{\mu_{\phi 1}^{*}, \ldots, \mu_{\phi, r_\phi}^{*}\}$.

The mean square estimation error is
\begin{equation}
\E\{(\bphi_t-\tilde{\bphi}_t)(\bphi_t-\tilde{\bphi}_t)' \}= \bGamma_0^\phi-\bGamma_0^\phi \bZ_\phi\bZ_\phi'\bGamma_0^\phi.
\label{eq:msephi}
\end{equation}
\end{proposition}
\begin{proof}
See Appendix \ref{app:prop3}.
\end{proof}
The linear combinations $\bef_{\phi t} = \bZ_\phi'\bx_t$ are the $r_\phi$ smooth generalized principal components \shortcite{Altissimo2010new}.
Notice that the solution can be equivalently written $\tilde{\bphi}_t = \bGamma_0^\phi\bZ_\phi \bZ_\phi'\bx_t$. The columns of the
matrix $\bZ_\phi$  are formed from the $r_\phi$ generalized eigenvectors of $(\bGamma_0^\phi, \bGamma_0^x)$,  satisfying
$\bGamma_0^\phi \bZ_\phi = \bGamma_0^x \bZ_\phi\bM_\phi^*,$ or equivalently $\bZ_\phi' \bGamma_0^\phi \bZ_\phi = \bM_\phi^*$ and $\bZ_\phi'\bGamma_0^x \bZ_\phi = \bI_r$.
For  $n\rightarrow\infty$, the sum space of $\bphi_t$ and $\bpsi_t$ is the space of the common component.
The consistency of $\tilde{\bphi}_t$ is considered in the next proposition:
\begin{proposition}\label{prop:4}
$\tilde{\bphi}_t\rightarrow_P \bphi_t$, as $n\rightarrow \infty$    and
$$\lim\limits_{n\rightarrow \infty} \tilde{\bphi}_t = \bGamma_0^x\bZ_\phi  \bZ_\phi'\bx_t.$$
\end{proposition}
\begin{proof}
See Appendix \ref{app:prop4}.
\end{proof}

The limiting solution can be written $\tilde{\bphi}_t = \bA_\phi \bef_{\phi t}$, where $\bef_{\phi t} = \bZ_\phi'\bx_t$ are the smooth generalized principal components,  or
$\tilde{\bphi}_t= \bP_\phi\bx_t$, where $\bP_\phi = \bGamma_0^x \bZ_\phi\bZ_\phi'$ is  idempotent projection matrix of rank $r_\phi$.

Notice that $\bI_n = \bQ_\phi+\bQ_\psi+\bQ_\xi$, where $\bQ_\phi$ was defined by Prop. \ref{prop:3}, $\bQ_\psi = \bM_x^{-\frac{1}{2}}\bS'\bGamma_0^\psi\bS\bM_x^{-\frac{1}{2}}$,  and
$\bQ_\xi = \bM_x^{-\frac{1}{2}}\bS'\bGamma_0^\xi\bS\bM_x^{-\frac{1}{2}}.$

As $n\rightarrow \infty$, $\bQ_\phi \rightarrow\bP^*_\phi =  \bZ_\phi^*\bZ_\phi^{*\prime}$, the orthogonal projection matrix into the subspace generated by $\bphi_t^*$. The subspace spanned by $\bphi_t$ is a linear transformation and the projection of $\bx_t^*$ onto this subspace is $\bS\bM_x^{\frac{1}{2}}\bZ_\phi^*\bZ_\phi^{*\prime}\bx_t^*$. We have $\bZ_\phi^{*\prime}\bx_t^* = \bZ_\phi'\bx_t$ and $\bS\bM_x^{\frac{1}{2}}\bZ_\phi^*=\bGamma_0^x\bZ_\phi$. Hence $\bGamma_0^x \bZ_\phi \bZ_\phi'\bx_t$ is the projection of $\bx_t$ onto the subspace spanned by $\bphi_t$ and $\bP_\phi = \bGamma_0^x\bZ_\phi \bZ_\phi'$ is the corresponding projection matrix.

\section{Monthly indicator of GDP growth    \label{sec:usacoin}}

The smooth generalized principal components are now used to construct the indicator of M2LR GDP growth at the monthly frequency.
The model for unobserved monthly GDP growth assumes that $\Delta y_t$ depends on the information available in real time only via the smooth generalized principal components, that is,
\begin{equation}
\Delta y_t = \mu + \bvartheta' \bef_{\phi t} + \epsilon_t, \;\;\; \epsilon_t \sim \mbox{WN}(0,\sigma^2). \label{eq:model}
\end{equation}
Given the estimate of $\bef_{\phi t}$ obtained in section (\ref{sec:est}), the vector of loadings $\bvartheta$ can be estimated from the  quarterly GDP growth rates,
$g_t\approxeq (1+L+L^2)^2 \Delta y_t$, which are observed every third month.  The approximation is due to the fact that the sum of the logarithms of monthly GDP is only a first order Taylor approximation of the logarithm of the sum of GDP across the three months representing the quarter.

Filtering both sides of (\ref{eq:model}) by $(1+L+L^2)^2$ and taking a systematic sample with step 3, yields the  approximate quarterly model
\begin{equation}
\label{eq:qonq}
g_\tau  =  9\mu + \bvartheta' \bF_\tau  + \epsilon_\tau^*, \;\;\tau = 1, 2, \ldots, \mathcal{T},
\end{equation}
where $\bF_\tau$ is a systematic sample of $\bF_t = (1+L+L^2)^2 \bef_{\phi t}$ and $\epsilon_{\tau}^*$ is the  MA(1) process $\epsilon_{\tau}^* = \eta^*_\tau +b \eta^*_{\tau-1}$, $b = .221$, $\eta^*_\tau \sim \mbox{WN}(0, 19\sigma^2)$.

As for the M2LR monthly indicator of growth  at the annual horizon,
letting $a_t = (1+L+L^2)(1+L+\cdots+L^{11})$, and denoting by  $a_\tau$ its systematic sample, we have
\begin{equation}
\label{eq:yoy}
a_\tau  =  36\mu + \bvartheta' \bF_\tau^*  + \epsilon_\tau^\dag, \;\;\tau = 1, 2, \ldots, \mathcal{T},
\end{equation}
where $\bF_\tau^*$ is a systematic sample of $\bF_t^* = (1+L+L^2)(1+L+\cdots+L^{11})\bef_{\phi t}$ and $\epsilon_\tau^\dag$ is an MA(4) process   with spectral density
$$\sigma_{\epsilon^\dag}(\theta) = 12\sum_{j=0}^2 \frac{\sin^2(3\theta_j/2)\sin^2(12\theta/2)}{\sin^4(\theta_j/2)}\sigma^2,\;\;\; \theta_j = \frac{\theta + 2\pi j}{3}, j=0,1,2.$$.

\section{Estimation in finite cross-sections \label{sec:est}}

Our methodology aims at producing estimates of the M2LR component of GDP growth at three horizons: the monthly horizon, via $\tilde{\Delta y}_t = \tilde{\mu}+\tilde{\bvartheta}' \tilde{\bef}_{\phi t}$, the quarterly horizon, via $\tilde{g}_t = 9\tilde{\mu}+\tilde{\bvartheta}' \tilde{F}_{t}$, and the annual horizon, via
$\tilde{a}_t = 36\tilde{\mu}+\tilde{\bvartheta}' \tilde{\bF}_t^*$.

The estimation of $\bphi_t$ requires selecting $\r_{\phi}$ and computing the generalized eigenvectors of the matrix $\tilde{\bGamma}_0^{\phi}$ in the metric of $\tilde{\bGamma}_0^x$. While the latter is estimated from the observed time series $\{\bx_t, t=1, \ldots, T\}$ by sample variance matrix $\tilde{\bGamma}_0^x = \frac{1}{T} \sum_t \bx_t\bx_t'$, the former requires inverting the spectral density estimate $\tilde{\bSigma}^\phi(\theta)$.

The spectral density of $\bx_t$ is estimated by the Bartlett estimator
$$\tilde{\bSigma}^x(\theta) = \frac{1}{2\pi} \sum_{k=-M_T}^{M_T}\left(1-\frac{|k|}{M_T+1}\right)\tilde{\bGamma}_{k}^xe^{-\imath \theta k},$$
where $\tilde{\bGamma}_k = \frac{1}{T}\sum_{t=k+1}^T \bx_t\bx_{t-k}$ is the sample crosscovariance matrix at lag $k$.
The window parameter is selected in the range $c T^{1/3}<M_T < cT^{1/2}$, given a positive constant $c$. The asymptotic theory for $\tilde{\bSigma}^x(\theta)$ is derived in \shortciteA{forni2017dynamic}, which shows its  uniform consistency with respect to $\theta$ as $T\rightarrow \infty$, under regularity conditions on the process $\bx_t$. In our illustrations we used $M_T = 20$, which corresponds to $0.75T^{1/2}$.  The spectrum is estimated at $2m+1=151$ equally spaced frequencies in the range $[-\pi, \pi]$, $\theta_h=2\pi h/(2m+1), h=-m, \ldots, m$.

The spectral density of the common component at $\theta_h$ is estimated by performing the eigendecomposition of $\tilde{\bSigma}^x(\theta_h) = \sum_{k=1}^n  \tilde{\lambda}_k(\theta_h) \tilde{\bp}_k(\theta_h)\tilde{\bp}_k^H(\theta_h)$ and setting
$$\tilde{\bSigma}^\chi(\theta_h) = \sum_{k=1}^q  \tilde{\lambda}_k(\theta_h)\tilde{\bp}_k(\theta_h)\tilde{\bar{\bp}}_k (\theta_h).
$$
The variance-covariance matrix of the common component is then estimated by the Riemann sum
$$\tilde{\bGamma}_0^\chi =  \frac{1}{2m+1}\sum_{h=-m}^{m}  \tilde{\bSigma}^\chi(\theta_h),$$
whereas that of $\bphi_t$ is estimated by summing across all $2m_c+1$ frequencies in the range $-\theta_c\leq\theta_h\leq \theta_c$, where $\theta_c=\pi/6$, i.e.,
$$\tilde{\bGamma}_0^\phi =  \frac{1}{2m_c+1}\sum_{h=-m_c}^{m_c}  \tilde{\bSigma}^\chi(\theta_h).$$
The generalized principal components are then estimated as  $\tilde{\phi}_t = \tilde{\bZ}_\phi'\bx_t$, where $\tilde{\bZ}_\phi$ are the generalized eigenvectors of $\tilde{\bGamma}^\phi_0$.

Finally, The loadings $\bvartheta$ are estimated by performing a band-spectrum regression \shortcite{engle1974band} of the observed quarterly growth rates $g_\tau$ on  $(1+L+L^2)^2 \tilde{\phi}_t$, sampled at quarterly intervals:  denoting the  Fourier transforms of $g_{\tau} $ and $\tilde{\bF}_{\tau}$, respectively by
$$J_g(\omega_j) = \frac{1}{\sqrt{2\pi \mathcal{T}}}\sum_{\tau=1}^\mathcal{T} g_{\tau} e^{-\imath \omega_j \tau}, \bJ_F(\omega_j) = \frac{1}{\sqrt{2\pi \mathcal{T}}}\sum_{\tau=1}^\mathcal{T} \tilde{\bF}_\tau e^{-\imath \omega_j \tau},$$
where $\omega_j=2\pi j/\mathcal{T}$ are the Fourier frequencies in the range $\Omega_c = (0,\pi/6)\cup (2\pi-\pi/6, 2\pi)$ and letting $S(\omega_j) = (1+b^2+2 b \cos \omega_j)$
$$\hat{\bvartheta} = \left(\sum_{j\in \Omega_c} \bJ_F(\omega_j) \bJ_F(\omega_j)^H/S(\omega_j)\right)^{-1}\sum_{j\in \Omega_c} \bJ_F(\omega_j) J_g(\omega)/S(\omega_j).$$

\subsection{The uncertainty in the smooth common components \label{sec:boot}}

The estimation error uncertainty for finite $n$ and $T$ can be evaluated by a parametric bootstrap method, proceeding along the following steps. By Theorem 9.4.4 in  \shortciteA{brillinger1981time}, the sampling distribution of the spectral density estimator $\hat{\bSigma}^x(\theta)$ can be approximated by a complex Wishart distribution with  $\nu = T/M_T$ degrees of freedom and scale matrix $\bSigma^x(\theta)/\nu$,  written   $\hat{\bSigma}^x(\theta)\sim W_{C}\left(\nu,\bSigma^x(\theta) \right)$. In view of the consistency rates in \shortciteA{forni2017dynamic}, we also considered setting the degrees of freedom $T^*= T/[M_T\log(M_T)]$.

Multiple independent draws are taken from  $\hat{\bSigma}^x(\theta_h)$. For each sample the dynamic eigenvalues and eigenvectors are computed, and conditional on $q$ and $r$ the smooth common component are estimated. Let $\hat{\bef}_{\phi t}^+ =\bZ_\phi^{+'}\bx_t$  denote a draw of  smooth generalized principal components based on the generalized eigenvectors of the matrix constructed from the simulated spectral density.

Letting  $\Re\left(\cdot\right)$ and $\Im\left(\cdot\right)$ denote respectively the real and complex part of the argument, and let $ \hat{\bSigma}^{x}_{\Re}(\theta_h) = \Re\left(\hat{\bSigma}^{x}(\theta_h)\right)$ and $ \hat{\bSigma}^{x}_{\Im}(j)  = \Im \left(\hat{\bSigma}^{x}(\theta_h)\right)$,  we form the matrix
 $\hat{\bS}(\theta_h) = \left(\begin{array}{c|c} \hat{\bSigma}^{x}_{\Re}(\theta_h) & \hat{\bSigma}^{x}_{\Im}(\theta_h) \\ \hline -\hat{\bSigma}^{x}_{\Im}(\theta_h) &  \hat{\bSigma}^{x}_{\Re}(\theta_h)\end{array}\right)$
where, for $h =-m, -m+1, \ldots, m-1, m,$  $\theta_h = 2\pi h/(2m+1) \in \left[-\pi;\pi\right]$; the resampling algorithm operates as follows.

For $b = 1,2,\dots,B$,
\begin{enumerate}
\item[(i)] draw $\hat{\bS}_b(\theta_h) \sim W(\nu, \nu^{-1}\hat{\bS}(\theta_h))$, $h =-m, -m+1, \ldots, m-1, m$;
\item[(ii)] generate $\hat{\bSigma}_b^{x}(\theta_h)  = \Re\left(\hat{\bS}_b^{+}(\theta_h) \right) -\imath  \Im \left(\hat{\bS}_b^{+}(\theta_h) \right)$;
\item[(iii)] compute the generalized eigenvectors $\hat{\bSigma}_b^{x}(\theta_h)$ and estimate the smooth generalized principal components, $\hat{\bef}_{\phi t}^{(b)} =\bZ_\phi^{(b)'}\bx_t$, as described in section \ref{sec:est};
\item[(iv)] using the procedure reported in Section \ref{sec:usacoin}, equations \eqref{eq:qonq}-\eqref{eq:yoy}, we construct a draw of
$\tilde{g}_{t}^{(b)} = 9\tilde{\mu}^{(b)}+\tilde{\bvartheta}^{(b)'} \tilde{\bF}_{t}^{(b)} \;\; \text{and} \;\; \tilde{a}_{t}^{(b)} = 36\tilde{\mu}^{(b)} + \tilde{\bvartheta}^{(b)'} \tilde{\bF}_t^{(b)'}$.
  \end{enumerate}

The draws $\tilde{g}_{t}^{(b)}$ and $\tilde{a}_{t}^{(b)}$, $b = 1, \ldots, B$, can be used to estimate the density nowcasts of GDP underlying growth, that we denote by $\tilde{f}_t^{g}(g_t|\mathcal{I}_t)$ and $\tilde{f}_t^{a}(a_t|\mathcal{I}_t)$, respectively, where $\mathcal{I}_t$ is the information set available at time $t$.

\section{Nowcasting the Medium-to-long run component of US GDP \label{sec:predval}}

\subsection{The data and the real-time simulation design}

The natural testbed for illustrating and assessing our methodology is the FRED-MD monthly database compiled by \shortciteA{mccracken2016fred}. The dataset  consists of monthly observations on 122 U.S. macroeconomic time series observed from January 1959 to December 2019, grouped into 8 categories, see Appendix \ref{app:series}. The original time series are subjected to a stationarity inducing transformation, according to \shortciteA{mccracken2016fred}. The target series, quarterly  GDP (Billions of Chained 2012 Dollars, seasonally adjusted) is also made available at the St. Luis FED Economic Data (https://fred.stlouisfed.org/series/GDPC1).

Our methodology produces monthly nowcasts  of the M2LR component of GDP growth $g_t$ (quarterly horizon) and $a_t$ (annual horizon), that will be referred to as US COIN (U.S. COincident INdicators) in the sequel. For measuring their accuracy, as already discussed in \shortciteA{Altissimo2010new}, we need to target a measure of  M2LR. The latter cannot be observed directly, but it can be estimated by an approximation to the ideal low-pass filter, the \shortciteA[BK]{baxter1999measuring} filter, to the GDP growth rates interpolated at the monthly frequency. The BK filter is a two-sided symmetric filter that results from truncating the ideal low-pass filter by considering 36 (3 years of monthly data) past and future observations, along with the reference time,  and rescaling the impulse response weights so that they sum up to 1. The filter is applied to the interpolated monthly proxies of $g_t$ and $a_t$, obtained by applying the Whittaker-Kotel'nikov-Shannon sampling theorem, see \shortciteA[Theorem 7.2.2]{partington1997interpolation}, to the extended series, as detailed in  Appendix \ref{sec:target}.  Appendix \ref{sec:target} also discusses the robustness of this measure, vis-a-vis an alternative estimate based on a model-based band-pass filter.

Figure \ref{fig:target} displays the  monthly  target measure of M2LR (red line), which provides the basis for the evaluation of our method, both for quarter-on-quarter (q-o-q) and the year-on-year (y-o-y) growth horizons. Notice that the target measure is an \emph{oracle}, as it uses also future observations. It is the value the M2LR component would take in month $t$, were GDP growth available for the future 3 years.

For evaluating the performance we perform a pseudo-real time rolling nowcasting experiment using, as a test sample the period January 1980-December 2018.  We use as initial training sample the first 241 observations (1960:3 - 1980:1) to estimate the M2LR at the end of the sample (January 1980), that we then compare to the benchmark estimate of the M2LR component of GDP growth. The training sample is then moved forward, by adding  one observation at the end and removing the first observation, so as to cover the period 1960:4 - 1980:2;  US COIN is re-estimated and a new nowcast for February 1980 is made available. Furthermore, the rank $r$ of the variance covariance matrix $\hat{\bGamma}_{0}$ and of the spectral density $\bSigma(\theta)$ is determined by means of the BN criterion proposed by \shortciteA{bai2002determining} and the HL criterion proposed by \shortciteA{hallin2007determining} respectively. This is iterated until we reach the end of the test sample (December 2018).

\begin{figure}[h]
\caption{\small{Monthly target measure of  M2LRG of U.S. GDP. Top panel: quarter-on-quarter (q-o-q) M2LR growth and GDP quarterly growth. Bottom panel: year-on-year (y-o-y) M2LR and yearly GDP growth}}
\centering
\includegraphics[width=12cm,height=6cm]{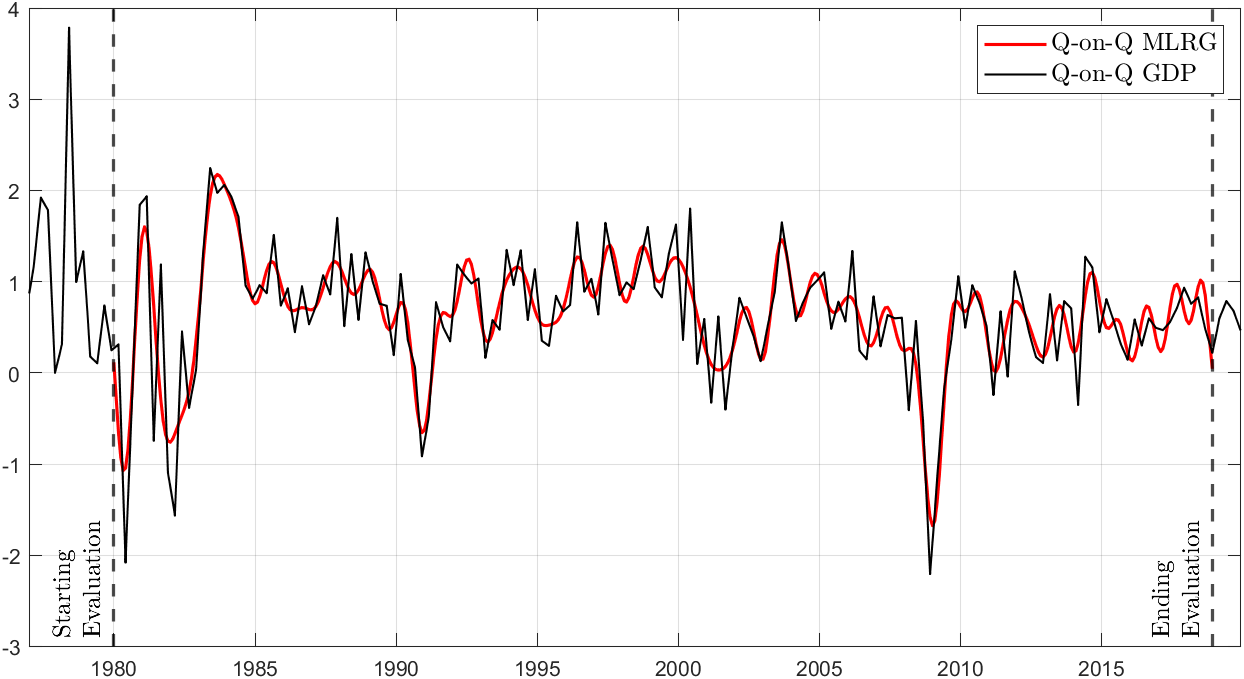}\\
\vspace{0.25cm}
\includegraphics[width=12cm,height=6cm]{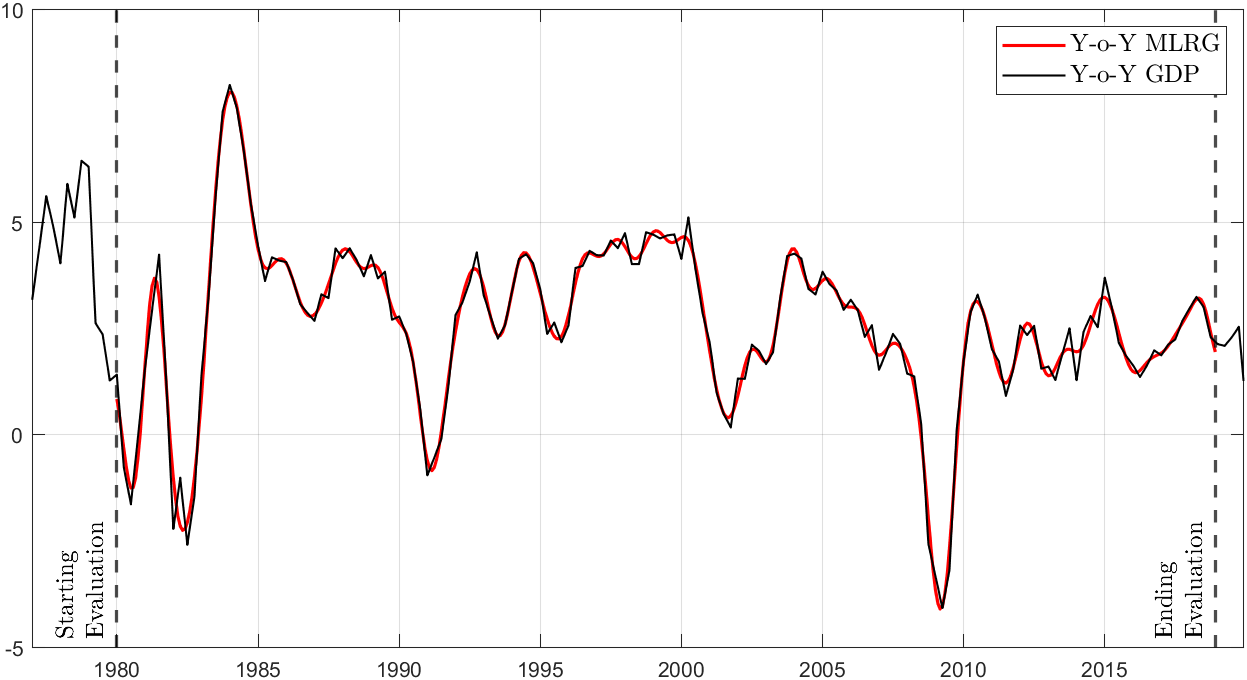}
\label{fig:target}
\end{figure}

For a comparative assessment, we could think of alternative ways to nowcasts  and forecasts the M2LR component of GDP growth, based on extending the quarterly GDP series by out-of-sample predictions. We envisaged three competitors.
\begin{itemize}
\item[BP] The BK filter is applied to the interpolated series after extending it by the sample mean of GDP growth.
\item[CF]  The quarterly GDP growth series is extended by AR(1) forecasts up to 12 quarters ahead. The AR(1) is estimated by least squares and the multistep ahead forecasts are obtained by the indirect method (chain rule). The application of the BK filter to the series extended by forecasts was proposed by \shortciteA[CF]{christiano2003band}.
\item[SW] The M2LR component is estimated according to the same US COIN methodology, but in section \ref{sec:usacoin} the smooth generalized principal components,  $\bef_{\phi t}$, are replaced by the standard principal components, $\bS_t \bx_t$, see \shortciteA{stock2002forecasting,stock2002macroeconomic}.
 \end{itemize}
To obtain the monthly nowcast for the methods BP and CF, the interpolated monthly GDP growth proxies of $g_t$ and $a_t$ are obtained by applying the Whittaker-Kotel'nikov-Shannon formula, see   Appendix \ref{sec:target}; hence, the BK filter with cut-off frequency $\pi/6$ is applied to the demeaned series, and finally the mean is added back to the filtered series to obtain the three nowcasts of the M2LR of GDP, which are compared to the USCOIN  nowcast.

\subsection{Empirical results}

For the sake of notation simplicity, we denote the q-on-q and y-on-y M2LRG targets at time $t$ by $\mathcal{C}_t$, and the estimate based on the rolling sample ending at time $s$ by $\hat{\mathcal{C}}_t(s)$.

The ability of the nowcast   $\hat{\mathcal{C}}_t(t)$, to approximate the target $\mathcal{C}_t$ is measured by the mean-square nowcast error (MSNE), $\sum_{t=1}^{T} [\hat{\mathcal{C}}_t(t) - \mathcal{C}_t]^2/T$, where $T$ is the length of the test sample.

We also consider the size of the revision after one month, as we move from one rolling sample to the next. This is measured by the mean-square revision error (MSRE),
$\sum_{t=1}^{T-1} [\hat{\mathcal{C}}_{t-1}(t) - \hat{\mathcal{C}}_{t-1}(t-1)]^2/T$.

Table \ref{tab:point_estimate} reports the values of the two statistics, as a fraction of the target variance,
$\sum_{t=1}^{T} (\mathcal{C}_{t} - \bar{\mathcal{C}})^2/T, \bar{\mathcal{C}} = \sum_t C_t/T$. The US COIN indicator outperforms its competitors at both horizons in terms of closeness to the target measure. Its MSNE for the q-o-q horizon is 12\% smaller than SW, the best performing competitor. The Diebold-Mariano \shortcite{diebold2002comparing} test of equal predictive accuracy under square loss leads to a rejection of the null at the  10\% significance level
For the y-o-y horizon it improves slightly over CF.
The size of the revision errors is relatively small for all indicators, except for  SW, which displays a larger MSRE.

\begin{table}[htbp]
  \centering
  \caption{Mean square nowcast error (MSNE) and mean square revision error (MSRE), relative to the target variance.}
    \label{tab:point_estimate}%
    \scalebox{0.85}{\begin{tabular}{rcc}
    \toprule
    \multicolumn{3}{c}{Panel A: M2LR q-on-q} \\
    \midrule
    \multicolumn{1}{c}{Indicator} & \multicolumn{1}{c}{Relative } & \multicolumn{1}{c}{Relative} \\
          & \multicolumn{1}{c}{MSNE} & \multicolumn{1}{c}{MSRE} \\
    \midrule
    \multicolumn{1}{l}{US COIN} & \multicolumn{1}{c}{0.333} & \multicolumn{1}{c}{0.185} \\
    \multicolumn{1}{l}{BP } & \multicolumn{1}{c}{0.479} & \multicolumn{1}{c}{0.168} \\
    \multicolumn{1}{l}{CF } & \multicolumn{1}{c}{0.406} & \multicolumn{1}{c}{0.203} \\
    \multicolumn{1}{l}{SW} & \multicolumn{1}{c}{0.378} & \multicolumn{1}{c}{0.354} \\
    \midrule
    \multicolumn{3}{c}{Panel B: M2LR y-o-y} \\
    \midrule
    \multicolumn{1}{l}{US COIN} & \multicolumn{1}{c}{0.291} & \multicolumn{1}{c}{0.108} \\
    \multicolumn{1}{l}{BP } & \multicolumn{1}{c}{0.509} & \multicolumn{1}{c}{0.103} \\
    \multicolumn{1}{l}{CF } & \multicolumn{1}{c}{0.301} & \multicolumn{1}{c}{0.063} \\
    \multicolumn{1}{l}{SW} & \multicolumn{1}{c}{0.315} & \multicolumn{1}{c}{0.201} \\
    \bottomrule
    \end{tabular}}\\%
\begin{center}
\end{center}
\end{table}%

\subsection{Assessing the US COIN nowcast uncertainty in real-time \label{sec:density}}

The bootstrap methodology described in section \ref{sec:boot} produces   draws $\tilde{g}_{t}^{(b)}$ and $\tilde{a}_{t}^{(b)}$, $b = 1, \ldots, B$, from the distribution of the M2LR indicators  conditional on a finite $n$ and finite $T$ panel of macroeconomic indicators and past quarterly GDP data.
Let  $\tilde{f}_t^{g}(g_t|\mathcal{I}_t)$ and $\tilde{f}_t^{a}(a_t|\mathcal{I}_t)$ denote respectively the density nowcasts of the q-o-q and y-o-y M2LR estimates, where $\mathcal{I}_t$ is the information set available at time $t$. The variability of the nowcasts is essentially due to the finite sample estimation error uncertainty concerning the smooth principal components, which is in turn ascribed to that concerning the spectral density matrix.

In other words, if we had an infinite cross-section observed for $T$ large, then $\bef_{\phi t}$ is asymptotically observed and the only parameter uncertainty would arise for the estimation of the loadings of GDP on the generalized principal components, denoted $\bvartheta$ in (\ref{eq:qonq}) and (\ref{eq:yoy}). However, for finite $n$ and $T$, $f_{\phi t}$ is affected by measurement error, so that  $\tilde{\bef}_t^{g}(g_t|\mathcal{I}_t)$ and $\tilde{f}_t^{a}(a_t|\mathcal{I}_t)$ enable to quantify the nowcast uncertainty, by considering the dispersion of the projection of GDP growth on the estimated smooth generalized principal components.

Figure \ref{fig:uncertainty} displays the evolution over time of the nowcasting densities of the q-o-q indicator (top figure) and y-o-y indicator (bottom figure). It is a fan chart with the different shades of red corresponding to the deciles of the distribution.
It should be remarked that the densities are estimated in real time using rolling samples of 20 years of data. If a recursively updated training set were used instead, the dispersion would be lower at the cost of larger biases, as the spectral density estimates would be less localized.

\begin{figure}[h]
\caption{Fan chart of US COIN nowcast distribution.\label{fig:uncertainty}}
\centering
\includegraphics[width=12cm,height=6cm]{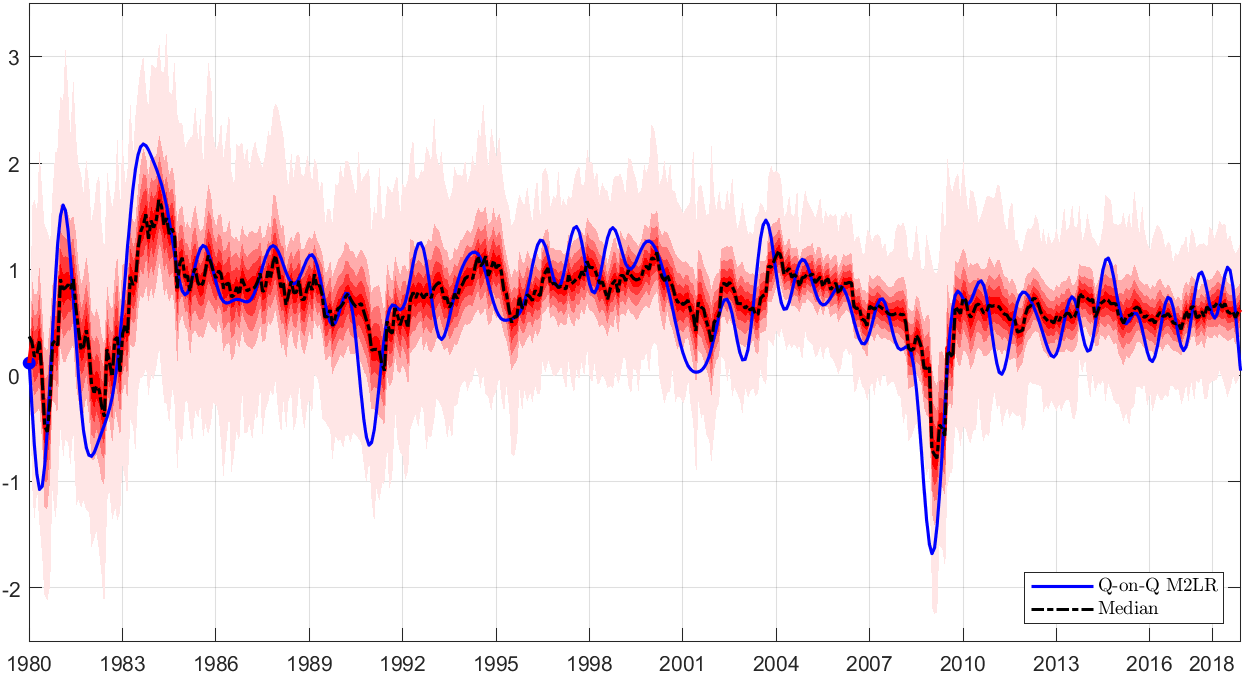}\\
\vspace{0.25cm}
\includegraphics[width=12cm,height=6cm]{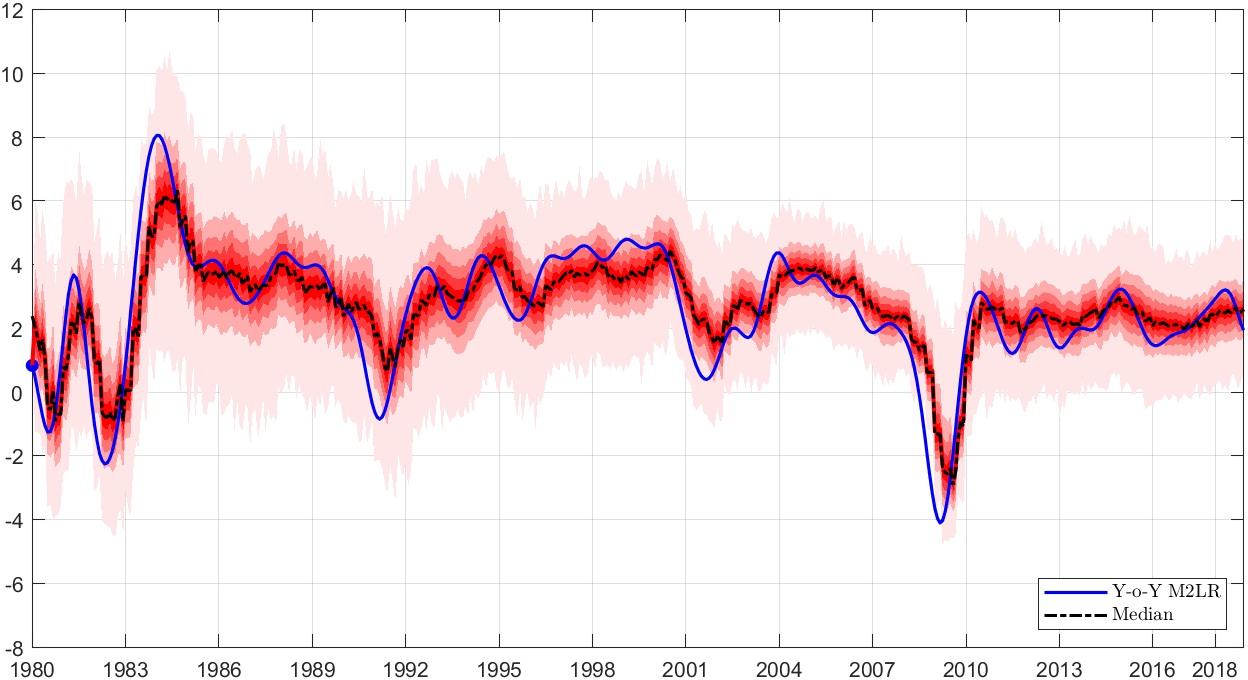}
\begin{center}
\end{center}
\end{figure}

To examine whether the nowcast densities are properly calibrated, we evaluate the cumulative distribution function of the probability integral transform of  $\mathcal{C}_t$,
$PIT_{t} = \int_{-\infty}^{\mathcal{C}_t} \tilde{f}_t^{i}(c|\mathcal{I}_t) d c, i=g,a$,
 and assess the departure from a standard uniform distribution according to \shortciteA{diebold1998evaluating}, using the critical values obtained in \shortciteA{rossi2019alternative}.
Figure \ref{fig:PIT} shows that the q-o-q nowcasts are well calibrated. The y-o-y nowcasts are marginally violating the 95\% bounds. This predictive failure is due to overprediction of GDP growth occurring around 2009 during the initial phase of the great recession; it can be considered a common feature of real time  macroeconomic forecasting.

\begin{figure}[bht]
\caption{Probability integral transform (PIT) of the  M2LR target $\mathcal{C}_t$: empirical cumulative distribution functions and 95\% confidence intervals under the null of uniformity. \label{fig:PIT} }
\centering
\includegraphics[width=5.5cm,height=5.5cm]{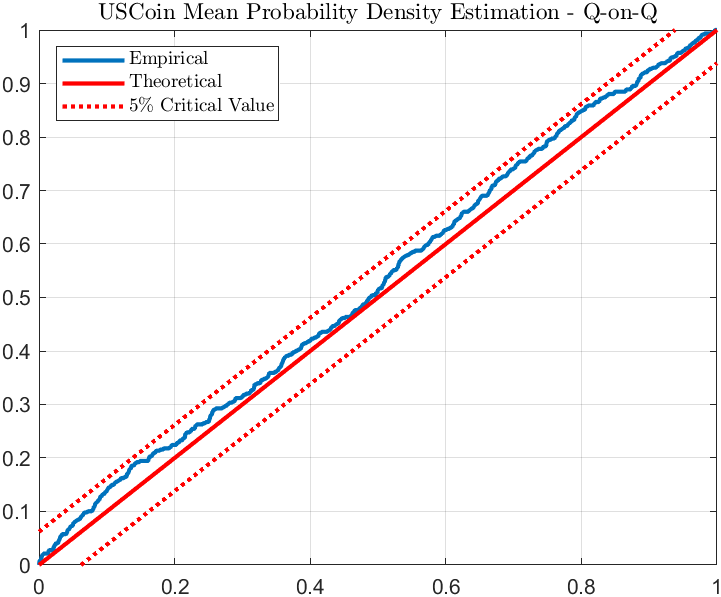}
\includegraphics[width=5.5cm,height=5.5cm]{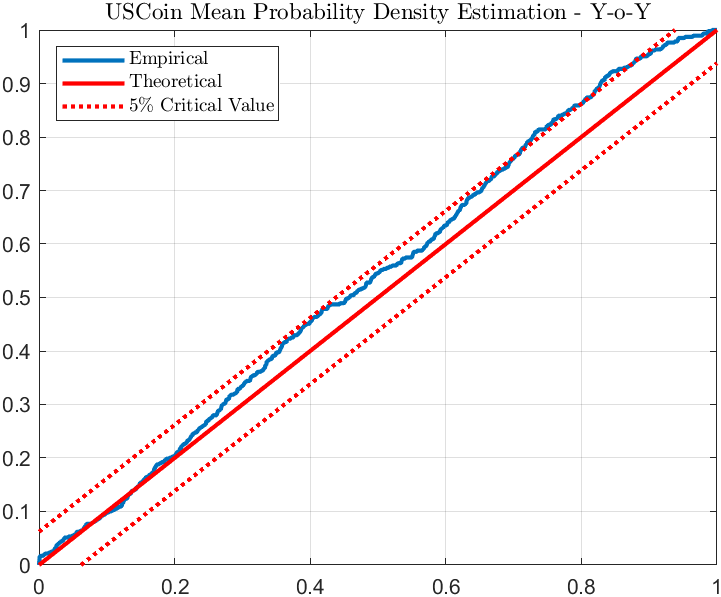}
\end{figure}
\section{Conclusions \label{sec:concl}}

The paper has considered the objective of estimating a monthly indicator of the M2LR component of GDP growth at the quarter-on-quarter and year-on-year horizons. The reference framework was the dynamic factor model, which was  used to distill the low frequency variation of a high-dimensional set of macroeconomic variables. After establishing the finite sample properties of the estimator of the smooth factors, the empirical application to nowcasting the underlying US GDP growth has illustrated that the smooth factors offer a sizable advantage in terms of  nowcasting accuracy towards a well defined target measure.

In our empirical exercise we have purposively ignored  the time series observations after March 2020, following a major structural break, related to the economic effects of the Covid-19 pandemic crisis. See  \shortciteA{ng2021modeling} for a discussion of those effects and a potential solution to the issue of estimating the smooth factors in the extended sample, covering the pandemic shock and the subsequent recovery. We leave this important topic to  future research.
\newpage
\bibliography{References}

\begin{thebibliography}{}

\bibitem [\protect \citeauthoryear {%
Altissimo%
, Cristadoro%
, Forni%
, Lippi%
\BCBL {}\ \BBA {} Veronese%
}{%
Altissimo%
\ \protect \BOthers {.}}{%
{\protect \APACyear {2010}}%
}]{%
Altissimo2010new}
\APACinsertmetastar {%
Altissimo2010new}%
\begin{APACrefauthors}%
Altissimo, F.%
, Cristadoro, R.%
, Forni, M.%
, Lippi, M.%
\BCBL {}\ \BBA {} Veronese, G.%
\end{APACrefauthors}%
\unskip\
\newblock
\APACrefYearMonthDay{2010}{}{}.
\newblock
{\BBOQ}\APACrefatitle {New Eurocoin: tracking economic growth in real time}
  {New eurocoin: tracking economic growth in real time}.{\BBCQ}
\newblock
\APACjournalVolNumPages{The Review of Economics and
  Statistics}{92}{4}{1024--1034}.
\PrintBackRefs{\CurrentBib}

\bibitem [\protect \citeauthoryear {%
Aruoba%
, Diebold%
, Nalewaik%
, Schorfheide%
\BCBL {}\ \BBA {} Song%
}{%
Aruoba%
\ \protect \BOthers {.}}{%
{\protect \APACyear {2016}}%
}]{%
aruoba2016improving}
\APACinsertmetastar {%
aruoba2016improving}%
\begin{APACrefauthors}%
Aruoba, S\BPBI B.%
, Diebold, F\BPBI X.%
, Nalewaik, J.%
, Schorfheide, F.%
\BCBL {}\ \BBA {} Song, D.%
\end{APACrefauthors}%
\unskip\
\newblock
\APACrefYearMonthDay{2016}{}{}.
\newblock
{\BBOQ}\APACrefatitle {Improving GDP measurement: A measurement-error
  perspective} {Improving gdp measurement: A measurement-error
  perspective}.{\BBCQ}
\newblock
\APACjournalVolNumPages{Journal of Econometrics}{191}{2}{384--397}.
\PrintBackRefs{\CurrentBib}

\bibitem [\protect \citeauthoryear {%
Baffigi%
, Golinelli%
\BCBL {}\ \BBA {} Parigi%
}{%
Baffigi%
\ \protect \BOthers {.}}{%
{\protect \APACyear {2004}}%
}]{%
baffigi2004bridge}
\APACinsertmetastar {%
baffigi2004bridge}%
\begin{APACrefauthors}%
Baffigi, A.%
, Golinelli, R.%
\BCBL {}\ \BBA {} Parigi, G.%
\end{APACrefauthors}%
\unskip\
\newblock
\APACrefYearMonthDay{2004}{}{}.
\newblock
{\BBOQ}\APACrefatitle {Bridge models to forecast the euro area GDP} {Bridge
  models to forecast the euro area gdp}.{\BBCQ}
\newblock
\APACjournalVolNumPages{International Journal of forecasting}{20}{3}{447--460}.
\PrintBackRefs{\CurrentBib}

\bibitem [\protect \citeauthoryear {%
Bai%
\ \BBA {} Ng%
}{%
Bai%
\ \BBA {} Ng%
}{%
{\protect \APACyear {2002}}%
}]{%
bai2002determining}
\APACinsertmetastar {%
bai2002determining}%
\begin{APACrefauthors}%
Bai, J.%
\BCBT {}\ \BBA {} Ng, S.%
\end{APACrefauthors}%
\unskip\
\newblock
\APACrefYearMonthDay{2002}{}{}.
\newblock
{\BBOQ}\APACrefatitle {Determining the number of factors in approximate factor
  models} {Determining the number of factors in approximate factor
  models}.{\BBCQ}
\newblock
\APACjournalVolNumPages{Econometrica}{70}{1}{191--221}.
\PrintBackRefs{\CurrentBib}

\bibitem [\protect \citeauthoryear {%
Ba{\'n}bura%
, Giannone%
, Modugno%
\BCBL {}\ \BBA {} Reichlin%
}{%
Ba{\'n}bura%
\ \protect \BOthers {.}}{%
{\protect \APACyear {2013}}%
}]{%
banbura2013now}
\APACinsertmetastar {%
banbura2013now}%
\begin{APACrefauthors}%
Ba{\'n}bura, M.%
, Giannone, D.%
, Modugno, M.%
\BCBL {}\ \BBA {} Reichlin, L.%
\end{APACrefauthors}%
\unskip\
\newblock
\APACrefYearMonthDay{2013}{}{}.
\newblock
{\BBOQ}\APACrefatitle {Now-casting and the real-time data flow} {Now-casting
  and the real-time data flow}.{\BBCQ}
\newblock
\BIn{} \APACrefbtitle {Handbook of economic forecasting} {Handbook of economic
  forecasting}\ (\BVOL~2, \BPGS\ 195--237).
\newblock
\APACaddressPublisher{}{Elsevier}.
\PrintBackRefs{\CurrentBib}

\bibitem [\protect \citeauthoryear {%
Baxter%
\ \BBA {} King%
}{%
Baxter%
\ \BBA {} King%
}{%
{\protect \APACyear {1999}}%
}]{%
baxter1999measuring}
\APACinsertmetastar {%
baxter1999measuring}%
\begin{APACrefauthors}%
Baxter, M.%
\BCBT {}\ \BBA {} King, R\BPBI G.%
\end{APACrefauthors}%
\unskip\
\newblock
\APACrefYearMonthDay{1999}{}{}.
\newblock
{\BBOQ}\APACrefatitle {Measuring business cycles: approximate band-pass filters
  for economic time series} {Measuring business cycles: approximate band-pass
  filters for economic time series}.{\BBCQ}
\newblock
\APACjournalVolNumPages{Review of Economics and Statistics}{81}{4}{575--593}.
\PrintBackRefs{\CurrentBib}

\bibitem [\protect \citeauthoryear {%
Brillinger%
}{%
Brillinger%
}{%
{\protect \APACyear {1981}}%
}]{%
brillinger1981time}
\APACinsertmetastar {%
brillinger1981time}%
\begin{APACrefauthors}%
Brillinger, D\BPBI R.%
\end{APACrefauthors}%
\unskip\
\newblock
\APACrefYear{1981}.
\newblock
\APACrefbtitle {Time series: data analysis and theory} {Time series: data
  analysis and theory}\ (\BVOL~36).
\newblock
\APACaddressPublisher{}{Siam}.
\PrintBackRefs{\CurrentBib}

\bibitem [\protect \citeauthoryear {%
Camacho%
\ \BBA {} Perez-Quiros%
}{%
Camacho%
\ \BBA {} Perez-Quiros%
}{%
{\protect \APACyear {2010}}%
}]{%
camacho2010introducing}
\APACinsertmetastar {%
camacho2010introducing}%
\begin{APACrefauthors}%
Camacho, M.%
\BCBT {}\ \BBA {} Perez-Quiros, G.%
\end{APACrefauthors}%
\unskip\
\newblock
\APACrefYearMonthDay{2010}{}{}.
\newblock
{\BBOQ}\APACrefatitle {Introducing the euro-sting: Short-term indicator of euro
  area growth} {Introducing the euro-sting: Short-term indicator of euro area
  growth}.{\BBCQ}
\newblock
\APACjournalVolNumPages{Journal of Applied Econometrics}{25}{4}{663--694}.
\PrintBackRefs{\CurrentBib}

\bibitem [\protect \citeauthoryear {%
Christiano%
\ \BBA {} Fitzgerald%
}{%
Christiano%
\ \BBA {} Fitzgerald%
}{%
{\protect \APACyear {2003}}%
}]{%
christiano2003band}
\APACinsertmetastar {%
christiano2003band}%
\begin{APACrefauthors}%
Christiano, L\BPBI J.%
\BCBT {}\ \BBA {} Fitzgerald, T\BPBI J.%
\end{APACrefauthors}%
\unskip\
\newblock
\APACrefYearMonthDay{2003}{}{}.
\newblock
{\BBOQ}\APACrefatitle {The band pass filter} {The band pass filter}.{\BBCQ}
\newblock
\APACjournalVolNumPages{International Economic Review}{44}{2}{435--465}.
\PrintBackRefs{\CurrentBib}

\bibitem [\protect \citeauthoryear {%
Diebold%
, Gunther%
\BCBL {}\ \BBA {} Tay%
}{%
Diebold%
\ \protect \BOthers {.}}{%
{\protect \APACyear {1998}}%
}]{%
diebold1998evaluating}
\APACinsertmetastar {%
diebold1998evaluating}%
\begin{APACrefauthors}%
Diebold, F\BPBI X.%
, Gunther, T.%
\BCBL {}\ \BBA {} Tay, A.%
\end{APACrefauthors}%
\unskip\
\newblock
\APACrefYearMonthDay{1998}{}{}.
\newblock
{\BBOQ}\APACrefatitle {Evaluating density forecasts, with evaluation to risk
  management} {Evaluating density forecasts, with evaluation to risk
  management}.{\BBCQ}
\newblock
\APACjournalVolNumPages{International Economic Review}{}{}{}.
\PrintBackRefs{\CurrentBib}

\bibitem [\protect \citeauthoryear {%
Diebold%
\ \BBA {} Mariano%
}{%
Diebold%
\ \BBA {} Mariano%
}{%
{\protect \APACyear {2002}}%
}]{%
diebold2002comparing}
\APACinsertmetastar {%
diebold2002comparing}%
\begin{APACrefauthors}%
Diebold, F\BPBI X.%
\BCBT {}\ \BBA {} Mariano, R\BPBI S.%
\end{APACrefauthors}%
\unskip\
\newblock
\APACrefYearMonthDay{2002}{}{}.
\newblock
{\BBOQ}\APACrefatitle {Comparing predictive accuracy} {Comparing predictive
  accuracy}.{\BBCQ}
\newblock
\APACjournalVolNumPages{Journal of Business \& economic
  statistics}{20}{1}{134--144}.
\PrintBackRefs{\CurrentBib}

\bibitem [\protect \citeauthoryear {%
Durbin%
\ \BBA {} Koopman%
}{%
Durbin%
\ \BBA {} Koopman%
}{%
{\protect \APACyear {2012}}%
}]{%
durbin2012time}
\APACinsertmetastar {%
durbin2012time}%
\begin{APACrefauthors}%
Durbin, J.%
\BCBT {}\ \BBA {} Koopman, S\BPBI J.%
\end{APACrefauthors}%
\unskip\
\newblock
\APACrefYear{2012}.
\newblock
\APACrefbtitle {Time series analysis by state space methods} {Time series
  analysis by state space methods}\ (\BVOL~38).
\newblock
\APACaddressPublisher{}{Oxford University Press}.
\PrintBackRefs{\CurrentBib}

\bibitem [\protect \citeauthoryear {%
Engle%
}{%
Engle%
}{%
{\protect \APACyear {1974}}%
}]{%
engle1974band}
\APACinsertmetastar {%
engle1974band}%
\begin{APACrefauthors}%
Engle, R\BPBI F.%
\end{APACrefauthors}%
\unskip\
\newblock
\APACrefYearMonthDay{1974}{}{}.
\newblock
{\BBOQ}\APACrefatitle {Band spectrum regression} {Band spectrum
  regression}.{\BBCQ}
\newblock
\APACjournalVolNumPages{International Economic Review}{}{}{1--11}.
\PrintBackRefs{\CurrentBib}

\bibitem [\protect \citeauthoryear {%
Forni%
, Giannone%
, Lippi%
\BCBL {}\ \BBA {} Reichlin%
}{%
Forni%
\ \protect \BOthers {.}}{%
{\protect \APACyear {2009}}%
}]{%
forni2009opening}
\APACinsertmetastar {%
forni2009opening}%
\begin{APACrefauthors}%
Forni, M.%
, Giannone, D.%
, Lippi, M.%
\BCBL {}\ \BBA {} Reichlin, L.%
\end{APACrefauthors}%
\unskip\
\newblock
\APACrefYearMonthDay{2009}{}{}.
\newblock
{\BBOQ}\APACrefatitle {Opening the black box: Structural factor models with
  large cross sections} {Opening the black box: Structural factor models with
  large cross sections}.{\BBCQ}
\newblock
\APACjournalVolNumPages{Econometric Theory}{25}{5}{1319--1347}.
\PrintBackRefs{\CurrentBib}

\bibitem [\protect \citeauthoryear {%
Forni%
, Hallin%
, Lippi%
\BCBL {}\ \BBA {} Reichlin%
}{%
Forni%
\ \protect \BOthers {.}}{%
{\protect \APACyear {2000}}%
}]{%
forni2000generalized}
\APACinsertmetastar {%
forni2000generalized}%
\begin{APACrefauthors}%
Forni, M.%
, Hallin, M.%
, Lippi, M.%
\BCBL {}\ \BBA {} Reichlin, L.%
\end{APACrefauthors}%
\unskip\
\newblock
\APACrefYearMonthDay{2000}{}{}.
\newblock
{\BBOQ}\APACrefatitle {The generalized dynamic-factor model: Identification and
  estimation} {The generalized dynamic-factor model: Identification and
  estimation}.{\BBCQ}
\newblock
\APACjournalVolNumPages{Review of Economics and Statistics}{82}{4}{540--554}.
\PrintBackRefs{\CurrentBib}

\bibitem [\protect \citeauthoryear {%
Forni%
, Hallin%
, Lippi%
\BCBL {}\ \BBA {} Reichlin%
}{%
Forni%
\ \protect \BOthers {.}}{%
{\protect \APACyear {2005}}%
}]{%
forni2005generalized}
\APACinsertmetastar {%
forni2005generalized}%
\begin{APACrefauthors}%
Forni, M.%
, Hallin, M.%
, Lippi, M.%
\BCBL {}\ \BBA {} Reichlin, L.%
\end{APACrefauthors}%
\unskip\
\newblock
\APACrefYearMonthDay{2005}{}{}.
\newblock
{\BBOQ}\APACrefatitle {The generalized dynamic factor model: one-sided
  estimation and forecasting} {The generalized dynamic factor model: one-sided
  estimation and forecasting}.{\BBCQ}
\newblock
\APACjournalVolNumPages{Journal of the American Statistical
  Association}{100}{471}{830--840}.
\PrintBackRefs{\CurrentBib}

\bibitem [\protect \citeauthoryear {%
Forni%
, Hallin%
, Lippi%
\BCBL {}\ \BBA {} Zaffaroni%
}{%
Forni%
\ \protect \BOthers {.}}{%
{\protect \APACyear {2017}}%
}]{%
forni2017dynamic}
\APACinsertmetastar {%
forni2017dynamic}%
\begin{APACrefauthors}%
Forni, M.%
, Hallin, M.%
, Lippi, M.%
\BCBL {}\ \BBA {} Zaffaroni, P.%
\end{APACrefauthors}%
\unskip\
\newblock
\APACrefYearMonthDay{2017}{}{}.
\newblock
{\BBOQ}\APACrefatitle {Dynamic factor models with infinite-dimensional factor
  space: Asymptotic analysis} {Dynamic factor models with infinite-dimensional
  factor space: Asymptotic analysis}.{\BBCQ}
\newblock
\APACjournalVolNumPages{Journal of Econometrics}{199}{1}{74--92}.
\PrintBackRefs{\CurrentBib}

\bibitem [\protect \citeauthoryear {%
Foroni%
\ \BBA {} Marcellino%
}{%
Foroni%
\ \BBA {} Marcellino%
}{%
{\protect \APACyear {2014}}%
}]{%
foroni2014comparison}
\APACinsertmetastar {%
foroni2014comparison}%
\begin{APACrefauthors}%
Foroni, C.%
\BCBT {}\ \BBA {} Marcellino, M.%
\end{APACrefauthors}%
\unskip\
\newblock
\APACrefYearMonthDay{2014}{}{}.
\newblock
{\BBOQ}\APACrefatitle {A comparison of mixed frequency approaches for
  nowcasting Euro area macroeconomic aggregates} {A comparison of mixed
  frequency approaches for nowcasting euro area macroeconomic
  aggregates}.{\BBCQ}
\newblock
\APACjournalVolNumPages{International Journal of Forecasting}{30}{3}{554--568}.
\PrintBackRefs{\CurrentBib}

\bibitem [\protect \citeauthoryear {%
Frale%
, Marcellino%
, Mazzi%
\BCBL {}\ \BBA {} Proietti%
}{%
Frale%
\ \protect \BOthers {.}}{%
{\protect \APACyear {2011}}%
}]{%
frale2011euromind}
\APACinsertmetastar {%
frale2011euromind}%
\begin{APACrefauthors}%
Frale, C.%
, Marcellino, M.%
, Mazzi, G\BPBI L.%
\BCBL {}\ \BBA {} Proietti, T.%
\end{APACrefauthors}%
\unskip\
\newblock
\APACrefYearMonthDay{2011}{}{}.
\newblock
{\BBOQ}\APACrefatitle {EUROMIND: a monthly indicator of the euro area economic
  conditions} {Euromind: a monthly indicator of the euro area economic
  conditions}.{\BBCQ}
\newblock
\APACjournalVolNumPages{Journal of the Royal Statistical Society: Series A
  (Statistics in Society)}{174}{2}{439--470}.
\PrintBackRefs{\CurrentBib}

\bibitem [\protect \citeauthoryear {%
Franklin%
}{%
Franklin%
}{%
{\protect \APACyear {2000}}%
}]{%
franklin2000matrix}
\APACinsertmetastar {%
franklin2000matrix}%
\begin{APACrefauthors}%
Franklin, J\BPBI N.%
\end{APACrefauthors}%
\unskip\
\newblock
\APACrefYear{2000}.
\newblock
\APACrefbtitle {Matrix theory} {Matrix theory}.
\newblock
\APACaddressPublisher{}{Dover Publications}.
\PrintBackRefs{\CurrentBib}

\bibitem [\protect \citeauthoryear {%
Ghysels%
, Sinko%
\BCBL {}\ \BBA {} Valkanov%
}{%
Ghysels%
\ \protect \BOthers {.}}{%
{\protect \APACyear {2007}}%
}]{%
ghysels2007midas}
\APACinsertmetastar {%
ghysels2007midas}%
\begin{APACrefauthors}%
Ghysels, E.%
, Sinko, A.%
\BCBL {}\ \BBA {} Valkanov, R.%
\end{APACrefauthors}%
\unskip\
\newblock
\APACrefYearMonthDay{2007}{}{}.
\newblock
{\BBOQ}\APACrefatitle {MIDAS regressions: Further results and new directions}
  {Midas regressions: Further results and new directions}.{\BBCQ}
\newblock
\APACjournalVolNumPages{Econometric reviews}{26}{1}{53--90}.
\PrintBackRefs{\CurrentBib}

\bibitem [\protect \citeauthoryear {%
Giannone%
, Reichlin%
\BCBL {}\ \BBA {} Small%
}{%
Giannone%
\ \protect \BOthers {.}}{%
{\protect \APACyear {2008}}%
}]{%
giannone2008nowcasting}
\APACinsertmetastar {%
giannone2008nowcasting}%
\begin{APACrefauthors}%
Giannone, D.%
, Reichlin, L.%
\BCBL {}\ \BBA {} Small, D.%
\end{APACrefauthors}%
\unskip\
\newblock
\APACrefYearMonthDay{2008}{}{}.
\newblock
{\BBOQ}\APACrefatitle {Nowcasting: The real-time informational content of
  macroeconomic data} {Nowcasting: The real-time informational content of
  macroeconomic data}.{\BBCQ}
\newblock
\APACjournalVolNumPages{Journal of monetary economics}{55}{4}{665--676}.
\PrintBackRefs{\CurrentBib}

\bibitem [\protect \citeauthoryear {%
Hallin%
, Lippi%
, Barigozzi%
, Forni%
\BCBL {}\ \BBA {} Zaffaroni%
}{%
Hallin%
\ \protect \BOthers {.}}{%
{\protect \APACyear {2020}}%
}]{%
hallin2020time}
\APACinsertmetastar {%
hallin2020time}%
\begin{APACrefauthors}%
Hallin, M.%
, Lippi, M.%
, Barigozzi, M.%
, Forni, M.%
\BCBL {}\ \BBA {} Zaffaroni, P.%
\end{APACrefauthors}%
\unskip\
\newblock
\APACrefYear{2020}.
\newblock
\APACrefbtitle {Time Series in High Dimensions: The General Dynamic Factor
  Model} {Time series in high dimensions: The general dynamic factor model}.
\newblock
\APACaddressPublisher{}{World Scientific}.
\PrintBackRefs{\CurrentBib}

\bibitem [\protect \citeauthoryear {%
Hallin%
\ \BBA {} Li{\v{s}}ka%
}{%
Hallin%
\ \BBA {} Li{\v{s}}ka%
}{%
{\protect \APACyear {2007}}%
}]{%
hallin2007determining}
\APACinsertmetastar {%
hallin2007determining}%
\begin{APACrefauthors}%
Hallin, M.%
\BCBT {}\ \BBA {} Li{\v{s}}ka, R.%
\end{APACrefauthors}%
\unskip\
\newblock
\APACrefYearMonthDay{2007}{}{}.
\newblock
{\BBOQ}\APACrefatitle {Determining the number of factors in the general dynamic
  factor model} {Determining the number of factors in the general dynamic
  factor model}.{\BBCQ}
\newblock
\APACjournalVolNumPages{Journal of the American Statistical
  Association}{102}{478}{603--617}.
\PrintBackRefs{\CurrentBib}

\bibitem [\protect \citeauthoryear {%
Kuzin%
, Marcellino%
\BCBL {}\ \BBA {} Schumacher%
}{%
Kuzin%
\ \protect \BOthers {.}}{%
{\protect \APACyear {2013}}%
}]{%
kuzin2013pooling}
\APACinsertmetastar {%
kuzin2013pooling}%
\begin{APACrefauthors}%
Kuzin, V.%
, Marcellino, M.%
\BCBL {}\ \BBA {} Schumacher, C.%
\end{APACrefauthors}%
\unskip\
\newblock
\APACrefYearMonthDay{2013}{}{}.
\newblock
{\BBOQ}\APACrefatitle {Pooling versus model selection for nowcasting GDP with
  many predictors: Empirical evidence for six industrialized countries}
  {Pooling versus model selection for nowcasting gdp with many predictors:
  Empirical evidence for six industrialized countries}.{\BBCQ}
\newblock
\APACjournalVolNumPages{Journal of Applied Econometrics}{28}{3}{392--411}.
\PrintBackRefs{\CurrentBib}

\bibitem [\protect \citeauthoryear {%
Mariano%
\ \BBA {} Murasawa%
}{%
Mariano%
\ \BBA {} Murasawa%
}{%
{\protect \APACyear {2003}}%
}]{%
mariano2003new}
\APACinsertmetastar {%
mariano2003new}%
\begin{APACrefauthors}%
Mariano, R\BPBI S.%
\BCBT {}\ \BBA {} Murasawa, Y.%
\end{APACrefauthors}%
\unskip\
\newblock
\APACrefYearMonthDay{2003}{}{}.
\newblock
{\BBOQ}\APACrefatitle {A new coincident index of business cycles based on
  monthly and quarterly series} {A new coincident index of business cycles
  based on monthly and quarterly series}.{\BBCQ}
\newblock
\APACjournalVolNumPages{Journal of applied Econometrics}{18}{4}{427--443}.
\PrintBackRefs{\CurrentBib}

\bibitem [\protect \citeauthoryear {%
Marshall%
, Olkin%
\BCBL {}\ \BBA {} Arnold%
}{%
Marshall%
\ \protect \BOthers {.}}{%
{\protect \APACyear {2010}}%
}]{%
marshall2010inequalities}
\APACinsertmetastar {%
marshall2010inequalities}%
\begin{APACrefauthors}%
Marshall, A\BPBI W.%
, Olkin, I.%
\BCBL {}\ \BBA {} Arnold, B\BPBI C.%
\end{APACrefauthors}%
\unskip\
\newblock
\APACrefYear{2010}.
\newblock
\APACrefbtitle {Inequalities: theory of majorization and its applications.
  Second Edition.} {Inequalities: theory of majorization and its applications.
  second edition.}
\newblock
\APACaddressPublisher{}{Springer}.
\PrintBackRefs{\CurrentBib}

\bibitem [\protect \citeauthoryear {%
McCracken%
\ \BBA {} Ng%
}{%
McCracken%
\ \BBA {} Ng%
}{%
{\protect \APACyear {2016}}%
}]{%
mccracken2016fred}
\APACinsertmetastar {%
mccracken2016fred}%
\begin{APACrefauthors}%
McCracken, M\BPBI W.%
\BCBT {}\ \BBA {} Ng, S.%
\end{APACrefauthors}%
\unskip\
\newblock
\APACrefYearMonthDay{2016}{}{}.
\newblock
{\BBOQ}\APACrefatitle {FRED-MD: A monthly database for macroeconomic research}
  {Fred-md: A monthly database for macroeconomic research}.{\BBCQ}
\newblock
\APACjournalVolNumPages{Journal of Business \& Economic
  Statistics}{34}{4}{574--589}.
\PrintBackRefs{\CurrentBib}

\bibitem [\protect \citeauthoryear {%
Ng%
}{%
Ng%
}{%
{\protect \APACyear {2021}}%
}]{%
ng2021modeling}
\APACinsertmetastar {%
ng2021modeling}%
\begin{APACrefauthors}%
Ng, S.%
\end{APACrefauthors}%
\unskip\
\newblock
\APACrefYearMonthDay{2021}{}{}.
\newblock
\APACrefbtitle {Modeling macroeconomic variations after COVID-19} {Modeling
  macroeconomic variations after covid-19}\ \APACbVolEdTR{}{\BTR{}}.
\newblock
\APACaddressInstitution{}{National Bureau of Economic Research}.
\PrintBackRefs{\CurrentBib}

\bibitem [\protect \citeauthoryear {%
Partington%
\ \protect \BOthers {.}}{%
Partington%
\ \protect \BOthers {.}}{%
{\protect \APACyear {1997}}%
}]{%
partington1997interpolation}
\APACinsertmetastar {%
partington1997interpolation}%
\begin{APACrefauthors}%
Partington, J\BPBI R.%
\BCBT {}\ \BOthersPeriod {.}
\end{APACrefauthors}%
\unskip\
\newblock
\APACrefYear{1997}.
\newblock
\APACrefbtitle {Interpolation, identification, and sampling} {Interpolation,
  identification, and sampling}\ (\BNUM~17).
\newblock
\APACaddressPublisher{}{Oxford University Press}.
\PrintBackRefs{\CurrentBib}

\bibitem [\protect \citeauthoryear {%
Proietti%
}{%
Proietti%
}{%
{\protect \APACyear {2008}}%
}]{%
proietti2008model}
\APACinsertmetastar {%
proietti2008model}%
\begin{APACrefauthors}%
Proietti, T.%
\end{APACrefauthors}%
\unskip\
\newblock
\APACrefYearMonthDay{2008}{}{}.
\newblock
{\BBOQ}\APACrefatitle {On the Model-Based Interpretation of Filters and the
  Reliability of Trend--Cycle Estimates} {On the model-based interpretation of
  filters and the reliability of trend--cycle estimates}.{\BBCQ}
\newblock
\APACjournalVolNumPages{Econometric Reviews}{28}{1-3}{186--208}.
\PrintBackRefs{\CurrentBib}

\bibitem [\protect \citeauthoryear {%
Rao%
}{%
Rao%
}{%
{\protect \APACyear {1974}}%
}]{%
rao1974projectors}
\APACinsertmetastar {%
rao1974projectors}%
\begin{APACrefauthors}%
Rao, C\BPBI R.%
\end{APACrefauthors}%
\unskip\
\newblock
\APACrefYearMonthDay{1974}{}{}.
\newblock
{\BBOQ}\APACrefatitle {Projectors, generalized inverses and the BLUE's}
  {Projectors, generalized inverses and the blue's}.{\BBCQ}
\newblock
\APACjournalVolNumPages{Journal of the Royal Statistical Society: Series B
  (Methodological)}{36}{3}{442--448}.
\PrintBackRefs{\CurrentBib}

\bibitem [\protect \citeauthoryear {%
Rao%
\ \BBA {} Mitra%
}{%
Rao%
\ \BBA {} Mitra%
}{%
{\protect \APACyear {1971}}%
}]{%
rao1971generalized}
\APACinsertmetastar {%
rao1971generalized}%
\begin{APACrefauthors}%
Rao, C\BPBI R.%
\BCBT {}\ \BBA {} Mitra, S\BPBI K.%
\end{APACrefauthors}%
\unskip\
\newblock
\APACrefYear{1971}.
\newblock
\APACrefbtitle {Generalized Inverse of Matrices and Its Applications}
  {Generalized inverse of matrices and its applications}.
\newblock
\APACaddressPublisher{}{Wiley}.
\PrintBackRefs{\CurrentBib}

\bibitem [\protect \citeauthoryear {%
Rossi%
\ \BBA {} Sekhposyan%
}{%
Rossi%
\ \BBA {} Sekhposyan%
}{%
{\protect \APACyear {2019}}%
}]{%
rossi2019alternative}
\APACinsertmetastar {%
rossi2019alternative}%
\begin{APACrefauthors}%
Rossi, B.%
\BCBT {}\ \BBA {} Sekhposyan, T.%
\end{APACrefauthors}%
\unskip\
\newblock
\APACrefYearMonthDay{2019}{}{}.
\newblock
{\BBOQ}\APACrefatitle {Alternative tests for correct specification of
  conditional predictive densities} {Alternative tests for correct
  specification of conditional predictive densities}.{\BBCQ}
\newblock
\APACjournalVolNumPages{Journal of Econometrics}{208}{2}{638--657}.
\PrintBackRefs{\CurrentBib}

\bibitem [\protect \citeauthoryear {%
Sayed%
\ \BBA {} Kailath%
}{%
Sayed%
\ \BBA {} Kailath%
}{%
{\protect \APACyear {2001}}%
}]{%
sayed2001survey}
\APACinsertmetastar {%
sayed2001survey}%
\begin{APACrefauthors}%
Sayed, A\BPBI H.%
\BCBT {}\ \BBA {} Kailath, T.%
\end{APACrefauthors}%
\unskip\
\newblock
\APACrefYearMonthDay{2001}{}{}.
\newblock
{\BBOQ}\APACrefatitle {A survey of spectral factorization methods} {A survey of
  spectral factorization methods}.{\BBCQ}
\newblock
\APACjournalVolNumPages{Numerical linear algebra with
  applications}{8}{6-7}{467--496}.
\PrintBackRefs{\CurrentBib}

\bibitem [\protect \citeauthoryear {%
Seber%
}{%
Seber%
}{%
{\protect \APACyear {2008}}%
}]{%
seber2008matrix}
\APACinsertmetastar {%
seber2008matrix}%
\begin{APACrefauthors}%
Seber, G\BPBI A.%
\end{APACrefauthors}%
\unskip\
\newblock
\APACrefYear{2008}.
\newblock
\APACrefbtitle {A matrix handbook for statisticians} {A matrix handbook for
  statisticians}\ (\BVOL~15).
\newblock
\APACaddressPublisher{}{John Wiley \& Sons}.
\PrintBackRefs{\CurrentBib}

\bibitem [\protect \citeauthoryear {%
Stock%
\ \BBA {} Watson%
}{%
Stock%
\ \BBA {} Watson%
}{%
{\protect \APACyear {2002}}%
{\protect \APACexlab {{\protect \BCnt {1}}}}}]{%
stock2002forecasting}
\APACinsertmetastar {%
stock2002forecasting}%
\begin{APACrefauthors}%
Stock, J\BPBI H.%
\BCBT {}\ \BBA {} Watson, M\BPBI W.%
\end{APACrefauthors}%
\unskip\
\newblock
\APACrefYearMonthDay{2002{\protect \BCnt {1}}}{}{}.
\newblock
{\BBOQ}\APACrefatitle {Forecasting using principal components from a large
  number of predictors} {Forecasting using principal components from a large
  number of predictors}.{\BBCQ}
\newblock
\APACjournalVolNumPages{Journal of the American Statistical
  Association}{97}{460}{1167--1179}.
\PrintBackRefs{\CurrentBib}

\bibitem [\protect \citeauthoryear {%
Stock%
\ \BBA {} Watson%
}{%
Stock%
\ \BBA {} Watson%
}{%
{\protect \APACyear {2002}}%
{\protect \APACexlab {{\protect \BCnt {2}}}}}]{%
stock2002macroeconomic}
\APACinsertmetastar {%
stock2002macroeconomic}%
\begin{APACrefauthors}%
Stock, J\BPBI H.%
\BCBT {}\ \BBA {} Watson, M\BPBI W.%
\end{APACrefauthors}%
\unskip\
\newblock
\APACrefYearMonthDay{2002{\protect \BCnt {2}}}{}{}.
\newblock
{\BBOQ}\APACrefatitle {Macroeconomic forecasting using diffusion indexes}
  {Macroeconomic forecasting using diffusion indexes}.{\BBCQ}
\newblock
\APACjournalVolNumPages{Journal of Business \& Economic
  Statistics}{20}{2}{147--162}.
\PrintBackRefs{\CurrentBib}

\bibitem [\protect \citeauthoryear {%
Whittle%
}{%
Whittle%
}{%
{\protect \APACyear {1963}}%
}]{%
whittle1963prediction}
\APACinsertmetastar {%
whittle1963prediction}%
\begin{APACrefauthors}%
Whittle, P.%
\end{APACrefauthors}%
\unskip\
\newblock
\APACrefYear{1963}.
\newblock
\APACrefbtitle {Prediction and regulation by linear least-square methods}
  {Prediction and regulation by linear least-square methods}.
\newblock
\APACaddressPublisher{}{English Universities Press}.
\PrintBackRefs{\CurrentBib}

\end{thebibliography}

\clearpage
\appendix

\section{Proofs of theorems}
\subsection{Proof of Proposition \ref{prop:1} \label{app:prop1}}

Premultiplying both sides of (\ref{eq:dfm}) by  $\bM_x^{-\frac{1}{2}}\bS'$ yields
$\bx_{t}^* = \bchi^*_t + \bxi_t^*$, with $\bx_t^* = \bM_x^{-\frac{1}{2}}\bS' \bx_t$, $\bchi_t^* = \bM_x^{-\frac{1}{2}}\bS' \bchi_t$, $\bxi_t^* = \bM_x^{-\frac{1}{2}}\bS' \bxi_t$.
The optimal linear predictor  of $\bchi^*_t$  based on $\bx_t^*$ is
$$\begin{array}{lll}
\hat{\bchi}_t^* &=& \Cov(\bchi_t^*, \bx_t^*)\Var^{-1}(\bx_t^*) \bx_t^*\\
             &=& \bQ_\chi  \bx_t^*,\\
\end{array}$$
since $\Cov(\bchi_t^*, \bx_t^*)=\bQ_\chi$, $\Var(\bx_t^*) = \bI_n$ and $\E(\bchi_t^*\bxi_t^{*\prime})=\0$. The best (least squares) rank $r$ approximation to $\hat{\bchi}_t^*$ is thus $\tilde{\bchi}_t^* = \bZ_\chi^*\bM_\chi^{*} \bZ_\chi^{*\prime} \bx_t^*$.
The optimal rank $r$ estimator of $\bchi_t =\bS\bM_x^{\frac{1}{2}}\bchi_t^*$ is then obtained by premultiplying $\tilde{\bchi}_t^*$ by $\bS\bM_x^{\frac{1}{2}}$, giving
$$\begin{array}{lll}
\tilde{\bchi}_t &=& \bS\bM_x^{\frac{1}{2}} \bZ_\chi^*\bM_\chi^{*} \bZ_\chi^{*\prime} \bx_t^*\\
             &=&\bS\bM_x^{\frac{1}{2}} \bZ_\chi^*\bM_\chi^{*} \bZ_\chi^{*\prime} \bM_x^{-\frac{1}{2}}\bS'\bx_t \\
             &=&\bS\bM_x^{\frac{1}{2}}( \bM_x^{\frac{1}{2}}\bS'\bS\bM_x^{-\frac{1}{2}}) \bZ_\chi^*\bM_\chi^{*} \bZ_\chi^{*\prime} \bM_x^{-\frac{1}{2}}\bS'\bx_t \\
             &=&(\bS\bM_x\bS') (\bS\bM_x^{-\frac{1}{2}} \bZ_\chi^*)\bM_\chi^{*} (\bZ_\chi^{*\prime} \bM_x^{-\frac{1}{2}} \bS')\bx_t \\
             &=&\bGamma_0^x \bZ_\chi\bM_\chi^{*} \bZ_\chi^{\prime} \bx_t. \\
\end{array}$$

Finally, the proof of (\ref{eq:msechi}) is direct in view of the identity $\bGamma_0^x \bZ_\chi\bM_\chi^{*} = \bGamma_0^{\chi}\bZ_\chi$.

\subsection{Proof of Proposition \ref{prop:2} \label{app:prop2}}
We start by proving that  $\bM_\chi^{*}\rightarrow \bI_r$:  $\bGamma_0^\xi=\bGamma_0^x-\bGamma_0^\chi>0$ and  Weyl's inequality imply
$$\mu_{j}^x-\mu_1^{\xi}\leq \mu_j^\chi\leq \mu_j^x, j=1, \ldots, r,$$
see \shortciteA[Section 6.7]{franklin2000matrix}. Since $\mu_{1}^\xi=O(1)$ as implied by Assumption 1.c,  $\mu_{j}^\chi/\mu_j^x\rightarrow 1$ as $n\rightarrow \infty$.
Moreover, $\bQ_\chi = \bI_n -\bQ_\xi$ implies that $\mu_j^{\chi*}\leq 1$. Then, by Proposition A.1.a. in \shortciteA{marshall2010inequalities},
$\mu_{j}^{\chi*}\geq \frac{\mu_{j}^{\chi}}{\mu_{1}^{x}}\geq \frac{\mu_{j}^{\chi}}{\mu_{j}^{x}}$, which converges to 1.

The mean square error  matrix in (\ref{eq:msechi}) converges to a zero matrix. This can be seen by projecting the estimation error in the space of the common component (the projection on the space spanned by the idiosyncratic component being identically zero): $\bZ_\chi'\E\{(\bchi_t-\tilde{\bchi}_t)(\bchi_t-\tilde{\bchi}_t)' \}\bZ_\chi= (\bI_r -\bM_\chi^*)\bM_\chi^*$, which converges to a zero matrix as $n\rightarrow \infty$, as implied by $\bM_\chi^{*}\rightarrow \bI_r$.

\subsection{Proof of Proposition \ref{prop:3} \label{app:prop3}}

The standardized principal components of $\bx_t$ decompose as follows: $\bx_t^* = \bphi_t^*+\bpsi_t^*+\bxi_t^*$, where $\bphi_t^* = \bM_x^{-\frac{1}{2}} \bS'\bphi_t$, and, similarly, $\bpsi_t^* = \bM_x^{-\frac{1}{2}} \bS'\bpsi_t$. The optimal linear predictor of $\bphi_t^*$ is $\hat{\bphi}_t^* = \bQ_{\phi} \bx_t^*$. The remaining steps of the proof are identical to those of Proposition \ref{prop:1}.

\subsection{Proof of Proposition \ref{prop:4} \label{app:prop4}}

By Assumption \ref{ass:2}, $r=r_\phi+r_\psi$. Now,   Weyl's inequalities, $\mu_j^\chi \geq \mu_l^\phi+\mu_m^\psi$, for $l+m-1\leq j \leq r$, and $\mu_j^\chi \leq \mu_i^\phi+\mu_k^\psi$, for $i+k-1\geq j$,   imply that the eigenvalues  of $\bGamma_0^\phi$ of order higher than $r_\phi$ must be bounded, whereas the low order eigenvalues diverge as $n\rightarrow 1$. As a consequence, $\bM_\phi^*$ converges to the identity matrix of order $r_\phi$.
The mean square error  matrix in (\ref{eq:msephi}) converges to a zero matrix. This can be seen by projecting the estimation error in the space of the common component (the projection on the space spanned by the idiosyncratic component being identically zero): $\bZ_\phi'\E\{(\bphi_t-\tilde{\bphi}_t)(\bphi_t-\tilde{\bphi}_t)' \}\bZ_\phi= (\bI_{r_\phi} -\bM_\phi^*)\bM_\phi^*$, which converges to a zero matrix for $n\rightarrow \infty$, as implied by $\bM_\phi^{*}\rightarrow \bI_{r_\phi}$.

\clearpage

\section{Construction of  the target variable: optimal interpolation and band-pass filtering \label{sec:target}}
\setcounter{equation}{0}
\renewcommand{\theequation}{B\arabic{equation}}
\setcounter{figure}{0}
\renewcommand{\thefigure}{B\arabic{equation}}
The monthly target variable is the historical M2LR component of GDP quarterly growth. This results from applying a \shortciteA{baxter1999measuring} type approximation
to the optimally interpolated monthly series obtained from the original quarterly growth series.
Let $y_t$ denote the logarithm of unobserved monthly GDP, and let us assume that the observed GDP quarterly growth rate is a systematic sample of the monthly process
\begin{equation}\begin{array}{lll} g_t &=& (1+L+L^2)y_t - (1+L+L^2)y_{t-3} \\
                    &=& (1+L+L^2)^2\Delta y_t.
\end{array}
\label{eq:gt}
\end{equation}
Clearly, this is valid only up to an approximation, since the logarithm of the sum of monthly GDP for three consecutive months does not equal the sum of the logarithms.
The observed time series is $g_{3\tau}, \tau =1,2, \ldots, \lfloor T/3\rfloor$, i.e., every third value of $g_t$ is available.

If $\sigma(\theta), \theta \in [0,\pi]$, denotes the spectral density of $\Delta y_t$, that of $g_t$ is
$9\frac{\sin^4(3\theta/2)}{\sin^4(\theta/2)}\sigma(\theta)$.
The spectral density of $g_\tau$, denoted $f_q(\theta)$, is obtained by applying the folding formula \shortcite[p. 179]{brillinger1981time}:
$$f_q(\theta) = 3\sum_{j=0}^2 \frac{\sin^4(3\theta_j/2)}{\sin^4(\theta_j/2)}\sigma(\theta_j),\;\;\; \theta_j = \frac{\theta + 2\pi j}{3}.$$

The Whittaker-Kotel'nikov-Shannon sampling theorem, see \shortciteA[Theorem 7.2.2]{partington1997interpolation} can be used for the optimal interpolation at the monthly frequency of the low-pass component of GDP growth. If $g_t$ is bandlimited to $\theta_0 \leq \pi/3$, meaning that $s(\theta)=0$ for $\theta \in(-\theta_0, \theta_0)$ mod $2\pi$, then no aliasing occurs, i.e., $f_q(\theta)=\sigma(\theta\frac{\theta_0}{\pi})$, and $g_t$, by the Whittaker-Kotel'nikov-Shannon sampling theorem, can be perfectly reconstructed from its values sampled at discrete time points spaced 3 time units apart:
$$g_t = \sum_\tau g_{3\tau} \frac{\sin(\theta_0 (t-3\tau))}{\theta_0(t-3\tau)}. $$

Given that we are interested in the low-pass component of $g_t$ with cut-off frequency $\theta_c=\pi/6$, corresponding to a period of one year, we could first interpolate $g_t$, obtaining
$$\hat{g}_t = \sum_\tau g_{3\tau} \frac{\sin(   \pi (t-3\tau)/3)}{\pi(t-3\tau)/3},$$
and subsequently apply the optimal low-pass filter with cutoff frequency $\pi/6$:
$$\tilde{g}_t = W(L) \hat{q}_t, \;\;\; W(L) = \frac{1}{6}+\sum_{j=1}^\infty \frac{\sin(\pi j/6)}{\pi j/6} (L^j+L^{-j}).$$

Equivalently, we could first evaluate the low-pass component of the quarterly growth rate
with  cut-off frequency $\theta_c=\pi/2$ corresponding to a period of one year for quarterly data,
$$\tilde{g}_{3\tau} = W_\tau(L) {q}_{3t},\;\; W_\tau(L) = \frac{1}{2}+\sum_{j=1}^\infty \frac{\sin(   \pi j/2)}{\pi j/2} (L^j+L^{-j}), $$
where now $L$ acts on the index $\tau$, i.e. $L^jg_{3\tau} = g_{3(\tau-j)}$, and subsequently interpolate $\tilde{g}_{3\tau}$ to the monthly frequency by
$$\tilde{g}_t = \sum_\tau \tilde{q}_{3\tau} \frac{\sin(   \pi (t-3\tau)/3)}{\pi(t-3\tau)/3}.$$
The target is obtained by truncating the optimal band-pass filter at lead and lag 36, which amount to using three years of past and future monthly observations.

A second relevant target is represented by the M2LR growth component with annual horizon.
The observed annual growth rate is a systematic sample of the monthly process
$$\begin{array}{lll}a_t  &=&  (1+L+L^2)y_t-(1+L+L^2)y_{t-12} \\
                    &=& (1+L+L^2)(1+L+\cdots+L^{11})\Delta y_t.
\end{array}$$
The observed time series is $a_{3\tau}, \tau =1,2, \ldots, \lfloor T/3\rfloor$, i.e., every third value of $a_t$ is available.

The spectral density of $W_t$ is $36\frac{\sin^2(3\theta/2)\sin^2(12\theta/2)}{\sin^4(\theta/2)}\sigma(\theta)$.
That of $a_\tau$, denoted $f_a(\theta)$, is again obtained by applying the folding formula:
$$f_a(\theta) = 12\sum_{j=0}^2 \frac{\sin^2(3\theta_j/2)\sin^2(12\theta/2)}{\sin^4(\theta_j/2)}\sigma(\theta_j),\;\;\; \theta_j = \frac{\theta + 2\pi j}{3}, j=0,1,2.$$

Applying the Whittaker-Kotel'nikov-Shannon sampling theorem and assuming that the process is band-limited with cut-off frequency $\theta_c=\pi/6$, the interpolant of  $a_t$ is
$$\hat{a}_t = \sum_\tau a_{3\tau} \frac{\sin(   \pi (t-3\tau)/3)}{\pi(t-3\tau)/3};$$
hence, we apply the optimal low-pass filter with cutoff frequency $\pi/6$ to $\hat{a}_t$, giving the M2LR component of the year-on-year growth,
$\tilde{a}_t = W(L) \hat{a}_t.$

\subsection{ARMA Model-based band-pass filtering \label{sec:mbfilt}}

An alternative way of constructing the target is via the model-based band-pass filtering approach proposed by \shortciteA{proietti2008model}.
Suppose that the monthly growth rate $g_t$ in (\ref{eq:gt}), which is only observed every third month,  has the following stationary and invertible ARMA($p$, $q$) representation
$\alpha(L) (g_t -\mu) = \beta(L)\epsilon_t$, $\epsilon_t \sim \WN(0, \sigma^2)$, $\alpha(L)=1-\alpha_1L-\cdots-\alpha_p L^p,$
$\beta(L)=1+\beta_1L+\cdots+\beta_q L^q.$

As in section \ref{sec:smooth}, consider the orthogonal decomposition of the reduced form white noise process
$$\epsilon_t = \frac{(1+L)^s}{\varphi(L)}\eta_t + \sqrt{\varsigma}\frac{(1-L)^s}{\varphi(L)}\zeta_t,$$
where $\eta_t\sim\WN(0,1), \zeta_t\sim\WN(0, 1),$ and  $\E(\eta_t\zeta_{t-j})=0, \forall j\in \mathbb{Z},$   $\varphi(L)$ is a scalar polynomial satisfying
\begin{equation}
\varphi(L)\varphi(L^{-1}) = (1+L)^s(1+L^{-1})^s+\varsigma (1-L)^s(1-L^{-1})^s, \label{eq:varphi}
\end{equation}
and $\varsigma$ is related to the cutoff frequency $\theta_c=\pi/6$ by $\varsigma = \left(\frac{1+\cos\theta_c}{1-\cos\theta_c}\right)^s.$
The factorization of the right hand side of (\ref{eq:varphi}) is made possible by the fact that its Fourier transform is strictly positive, see
\shortciteA{sayed2001survey}, who also survey algorithms for computing $\varphi(L)$ from $s$ and $\varsigma$.

We obtain the  decomposition of $g_t=\ell_t +  h_t$,  where $l_t$ is the low-pass component with cut-off frequency $\pi/6$ and $h_t$ is the high-pass component, which have respectively the following strictly non-invertible ARMA($p+s$, $q+s$) representations:
$$\ell_t = \mu + \frac{\beta(L)(1+L)^s}{\alpha(L)\varphi(L)}\eta_t,  \;\;\; h_t =  \frac{\beta(L)(1-L)^s}{\alpha(L)\varphi(L)}\eta_t, $$
(notice that the spectral density of $\ell_t$ is zero at the $\pi$ frequency, whereas that of $h_t$ has $s$ zeroes at the zero frequency).

The minimum mean square linear estimator of $\ell_t$ based on a doubly infinite sample $g_{t-j}, j= -\infty, -1,0,1,\ldots, \infty$ is $\hat{\ell}_t = w_\ell(L)g_t$, where the
Wiener-Kolmogorov filter for estimating the low-pass component is
$$ w_\ell(L) = \frac{(1+L)^s(1+L^{-1})^s}{\varphi(L)\varphi(L^{-1})},$$
see \shortciteA{whittle1963prediction}. The gain of the filter declines monotonically from 1 to zero as $\theta$ ranges from 0  to $\pi$ and $w(\theta_c)=1/2$.

Given a systematic sample of $g_t$ with only quarterly observations, the estimation of $\ell_t$ is carried out by representing the decomposition as a state space model and applying the Kalman filter and smoother, see \shortciteA{durbin2012time}. Hence, interpolation and filtering are performed by state space methods.

For U.S. GDP we assumed an AR(1) model for $g_t$ and set $s=6$. Figure \ref{fig:mbtarget} compares the model-based estimate of the low-pass component (solid blue line) with the Baxter and King estimate of the low-pass component computed on the interpolated series (which results from the sum of the smoothed estimates of $\ell_t$ and $h_t$). The estimates are virtually identical, the only difference being that the Kalman filter and smoother are also available at the beginning and at the end of the sample (first and last 36 time points).

\begin{figure}[bht]
\caption{Comparison of alternative targets for the M2LR component of GDP growth (1947:Q2-2019:Q4). Original quarterly time series (circles), monthly interpolated time series (dashed line), model-based estimate of the low-pass component (solid blue line), and Baxter and King estimate of the low-pass component (red solid line). \label{fig:mbtarget} }
\centering
\includegraphics[width=15cm]{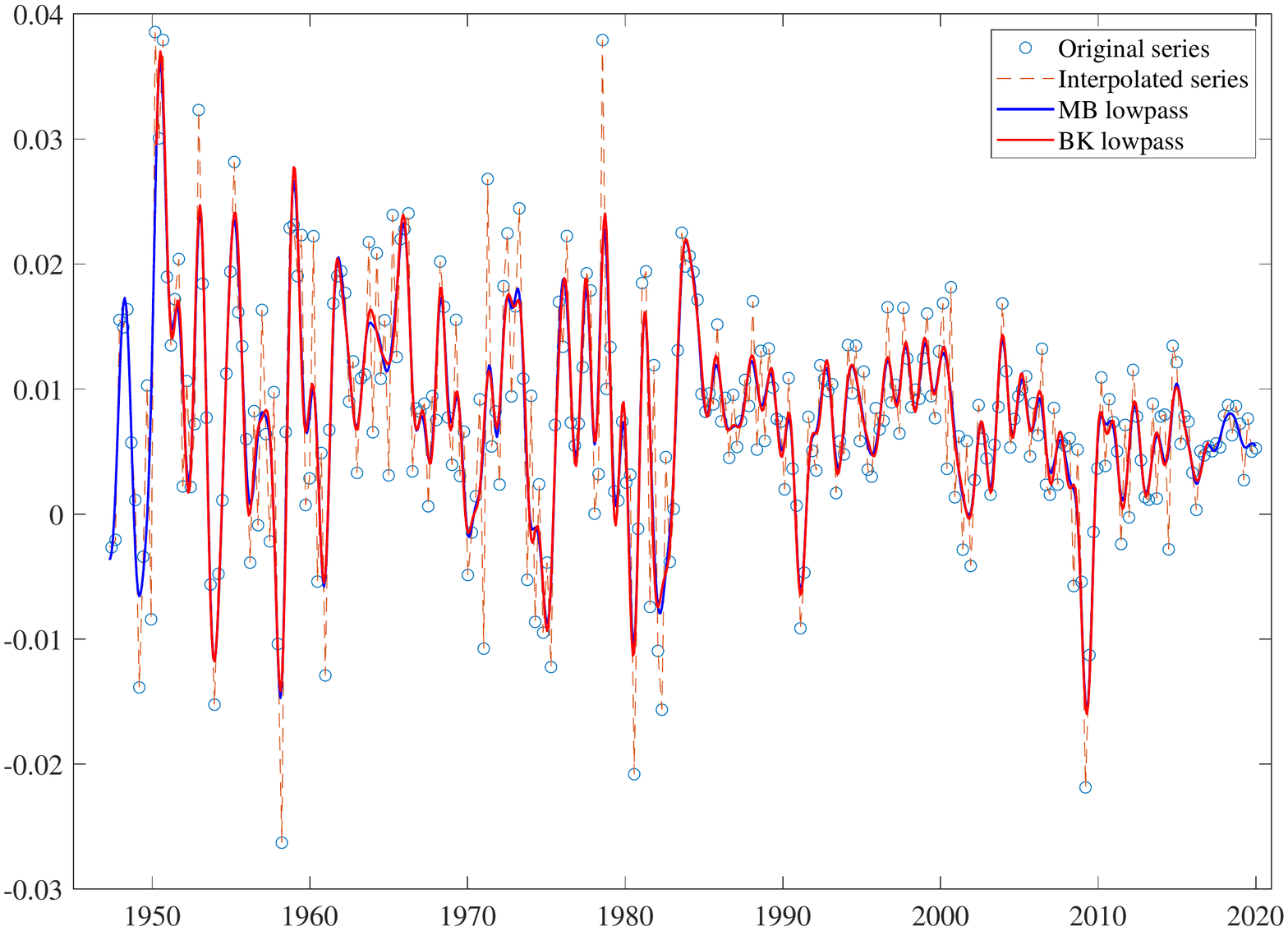}
\end{figure}

\clearpage
\section{List of Series \label{app:series}}
We report the list of series used in the empirical exercise. We use the same mnemonic group and transformation as described in McCracken and Ng. In particular, the column tcode denotes the following data transformation for a series $x_t$: (1) no transformation; (2) $\Delta x_t$; (3)$\Delta^2 x_t$; (4) $ log(x_t)$; (5) $\Delta log(x_t)$; (6) $\Delta^2 log(x_t)$. (7) $\Delta (x_t/x_{t-1} - 1.0)$. Series marked with * have not been included. The data used refer to the vintages \texttt{2022-01.csv} that can be downloaded from %

\noindent \url{https://research.stlouisfed.org/econ/mccracken/fred-databases/}.
{\footnotesize

\begin{longtable}{cclc}
\caption{Description of the dataset used in the empirical exercises} \label{tab:fred} \\

\hline \multicolumn{1}{c}{\textbf{ID}} & \multicolumn{1}{c}{\textbf{TCODE}} & \multicolumn{1}{c}{\textbf{Description}} & \multicolumn{1}{c}{\textbf{Group}} \\ \hline
\endfirsthead

\multicolumn{4}{c}%
{{\bfseries \tablename\ \thetable{} -- continued from previous page}} \\
\hline \multicolumn{1}{c}{\textbf{ID}} & \multicolumn{1}{c}{\textbf{TCODE}} & \multicolumn{1}{c}{\textbf{Description}} & \multicolumn{1}{c}{\textbf{Group}} \\ \hline
\endhead

\hline \multicolumn{4}{r}{{Continued on next page}} \\
\endfoot

\hline \hline
\endlastfoot

    1     & 5     & Real Personal Income & 1 \\
    2     & 5     & Real personal income ex transfer receipts & 1 \\
    3     & 5     & Real personal consumption expenditures & 4 \\
    4     & 5     & Real Manu.  and Trade Industries Sales & 4 \\
    5     & 5     & Retail and Food Services Sales & 4 \\
    6     & 5     & IP Index & 1 \\
    7     & 5     & IP: Final Products and Nonindustrial Supplies & 1 \\
    8     & 5     & IP: Final Products (Market Group) & 1 \\
    9     & 5     & IP: Consumer Goods & 1 \\
    10    & 5     & IP: Durable Consumer Goods & 1 \\
    11    & 5     & IP: Nondurable Consumer Goods & 1 \\
    12    & 5     & IP: Business Equipment & 1 \\
    13    & 5     & IP: Materials & 1 \\
    14    & 5     & IP: Durable Materials & 1 \\
    15    & 5     & IP: Nondurable Materials & 1 \\
    16    & 5     & IP: Manufacturing (SIC) & 1 \\
    17    & 5     & IP: Residential Utilities & 1 \\
    18    & 5     & IP: Fuels & 1 \\
    20    & 2     & Capacity Utilization:  Manufacturing & 1 \\
    21    & 2     & Help-Wanted Index for United States & 2 \\
    22    & 2     & Ratio of Help Wanted/No.  Unemployed & 2 \\
    23    & 5     & Civilian Labor Force & 2 \\
    24    & 5     & Civilian Employment & 2 \\
    25    & 2     & Civilian Unemployment Rate & 2 \\
    26    & 2     & Average Duration of Unemployment (Weeks) & 2 \\
    27    & 5     & Civilians Unemployed - Less Than 5 Weeks & 2 \\
    28    & 5     & Civilians Unemployed for 5-14 Weeks & 2 \\
    29    & 5     & Civilians Unemployed - 15 Weeks \& Over & 2 \\
    30    & 5     & Civilians Unemployed for 15-26 Weeks & 2 \\
    31    & 5     & Civilians Unemployed for 27 Weeks and Over & 2 \\
    32    & 5     & Initial Claims & 2 \\
    33    & 5     & All Employees:  Total nonfarm & 2 \\
    34    & 5     & All Employees:  Goods-Producing Industries & 2 \\
    35    & 5     & All Employees:  Mining and Logging:  Mining & 2 \\
    36    & 5     & All Employees:  Construction & 2 \\
    37    & 5     & All Employees:  Manufacturing & 2 \\
    38    & 5     & All Employees:  Durable goods & 2 \\
    39    & 5     & All Employees:  Nondurable goods & 2 \\
    40    & 5     & All Employees:  Service-Providing Industries & 2 \\
    41    & 5     & All Employees:  Trade, Transportation \& Utilities & 2 \\
    42    & 5     & All Employees:  Wholesale Trade & 2 \\
    43    & 5     & All Employees:  Retail Trade & 2 \\
    44    & 5     & All Employees:  Financial Activities & 2 \\
    45    & 5     & All Employees:  Government & 2 \\
    46    & 1     & Avg Weekly Hours :  Goods-Producing & 2 \\
    47    & 2     & Avg Weekly Overtime Hours :  Manufacturing & 2 \\
    48    & 1     & Avg Weekly Hours :  Manufacturing & 2 \\
    50    & 4     & Housing Starts:  Total New Privately Owned & 3 \\
    51    & 4     & Housing Starts, Northeast & 3 \\
    52    & 4     & Housing Starts, Midwest & 3 \\
    53    & 4     & Housing Starts, South & 3 \\
    54    & 4     & Housing Starts, West & 3 \\
    55    & 4     & New Private Housing Permits (SAAR) & 3 \\
    56    & 4     & New Private Housing Permits, Northeast (SAAR) & 3 \\
    57    & 4     & New Private Housing Permits, Midwest (SAAR) & 3 \\
    58    & 4     & New Private Housing Permits, South (SAAR) & 3 \\
    59    & 4     & New Private Housing Permits, West (SAAR) & 3 \\
    64    & 5     & New Orders for Consumer Goods* & 4 \\
    65    & 5     & New Orders for Durable Goods & 4 \\
    66    & 5     & New Orders for Nondefense Capital Goods* & 4 \\
    67    & 5     & Unfilled Orders for Durable Goods & 4 \\
    68    & 5     & Total Business Inventories & 4 \\
    69    & 2     & Total Business:  Inventories to Sales Ratio & 4 \\
    70    & 6     & M1 Money Stock & 5 \\
    71    & 6     & M2 Money Stock & 5 \\
    72    & 5     & Real M2 Money Stock & 5 \\
    73    & 6     & Monetary Base & 5 \\
    74    & 6     & Total Reserves of Depository Institutions & 5 \\
    75    & 7     & Reserves Of Depository Institutions & 5 \\
    76    & 6     & Commercial and Industrial Loans & 5 \\
    77    & 6     & Real Estate Loans at All Commercial Banks & 5 \\
    78    & 6     & Total Nonrevolving Credit & 5 \\
    79    & 2     & Nonrevolving consumer credit to Personal Income & 5 \\
    80    & 5     & S\&P's Common Stock Price Index: Composite & 8 \\
    81    & 5     & S\&P's Common Stock Price Index: Industrials & 8 \\
    82    & 2     & S\&P's Composite Common Stock: Dividend Yield & 8 \\
    83    & 5     & S\&P's Composite Common Stock: Price-Earnings Ratio & 8 \\
    84    & 2     & Effective Federal Funds Rate & 6 \\
    85    & 2     & 3-Month AA Financial Commercial Paper Rate & 6 \\
    86    & 2     & 3-Month Treasury Bill: & 6 \\
    87    & 2     & 6-Month Treasury Bill: & 6 \\
    88    & 2     & 1-Year Treasury Rate & 6 \\
    89    & 2     & 5-Year Treasury Rate & 6 \\
    90    & 2     & 10-Year Treasury Rate & 6 \\
    91    & 2     & Moody's Seasoned Aaa Corporate Bond Yield & 6 \\
    92    & 2     & Moody's Seasoned Baa Corporate Bond Yield & 6 \\
    93    & 1     & 3-Month Commercial Paper Minus FEDFUNDS & 6 \\
    94    & 1     & 3-Month Treasury C Minus FEDFUNDS & 6 \\
    95    & 1     & 6-Month Treasury C Minus FEDFUNDS & 6 \\
    96    & 1     & 1-Year Treasury C Minus FEDFUNDS & 6 \\
    97    & 1     & 5-Year Treasury C Minus FEDFUNDS & 6 \\
    98    & 1     & 10-Year Treasury C Minus FEDFUNDS & 6 \\
    99    & 1     & Moody's Aaa Corporate Bond Minus FEDFUNDS & 6 \\
    100   & 1     & Moody's Baa Corporate Bond Minus FEDFUNDS & 6 \\
    101   & 5     & Trade Weighted U.S. Dollar Index* & 6 \\
    102   & 5     & Switzerland / U.S. Foreign Exchange Rate & 6 \\
    103   & 5     & Japan / U.S. Foreign Exchange Rate & 6 \\
    104   & 5     & U.S. / U.K. Foreign Exchange Rate & 6 \\
    105   & 5     & Canada / U.S. Foreign Exchange Rate & 6 \\
    106   & 6     & PPI: Finished Goods & 7 \\
    107   & 6     & PPI: Finished Consumer Goods & 7 \\
    108   & 6     & PPI: Intermediate Materials & 7 \\
    109   & 6     & PPI: Crude Materials & 7 \\
    110   & 6     & Crude Oil, spliced WTI and Cushing & 7 \\
    111   & 6     & PPI: Metals and metal products: & 7 \\
    113   & 6     & CPI : All Items & 7 \\
    114   & 6     & CPI : Apparel & 7 \\
    115   & 6     & CPI : Transportation & 7 \\
    116   & 6     & CPI : Medical Care & 7 \\
    117   & 6     & CPI : Commodities & 7 \\
    118   & 6     & CPI : Durables & 7 \\
    119   & 6     & CPI : Services & 7 \\
    120   & 6     & CPI : All Items Less Food & 7 \\
    121   & 6     & CPI : All items less shelter & 7 \\
    122   & 6     & CPI : All items less medical care & 7 \\
    123   & 6     & Personal Cons.  Expend.:  Chain Index & 7 \\
    124   & 6     & Personal Cons.  Exp:  Durable goods & 7 \\
    125   & 6     & Personal Cons.  Exp:  Nondurable goods & 7 \\
    126   & 6     & Personal Cons.  Exp:  Services & 7 \\
    127   & 6     & Avg Hourly Earnings :  Goods-Producing & 2 \\
    128   & 6     & Avg Hourly Earnings :  Construction & 2 \\
    129   & 6     & Avg Hourly Earnings :  Manufacturing & 2 \\
    130   & 2     & Consumer Sentiment Index* & 4 \\
    132   & 6     & Consumer Motor Vehicle Loans Outstanding & 5 \\
    133   & 6     & Total Consumer Loans and Leases Outstanding & 5 \\
    134   & 6     & Securities in Bank Credit at All Commercial Banks & 5 \\
    135   & 1     & VIX*  & 8 \\

\end{longtable}%
}

\normalsize

\end{document}